\documentclass[10pt,twocolumn,letterpaper]{article}

\usepackage[pagenumbers]{cvpr} 

\usepackage{graphicx}

\usepackage{amsmath, amssymb}
\usepackage{amsfonts}
\usepackage{bm}

\usepackage{amsthm}
\usepackage{thmtools,thm-restate}
\newtheorem{theorem}{Theorem}

\newtheorem{lemma}{Lemma}
\newtheorem{definition}{Definition}

\usepackage{algorithm}
\usepackage{algpseudocode}

\usepackage{array}

\usepackage{caption}
\usepackage{subcaption}
\usepackage{booktabs}
\usepackage{afterpage}

\usepackage[moderate, paragraphs=normal, lists=normal, charwidths=normal, mathdisplays=normal]{savetrees}

\usepackage{enumitem}

\usepackage{url}

\algnewcommand{\LineComment}[1]{\State \(\triangleright\) #1}

\definecolor{cvprblue}{rgb}{0.21,0.49,0.74}
\usepackage[pagebackref,breaklinks,colorlinks,allcolors=cvprblue]{hyperref}

\def\paperID{12054} 
\def\confName{CVPR}
\def\confYear{2025}

\title{PatchDEMUX: A Certifiably Robust Framework for Multi-label Classifiers \\ Against Adversarial Patches}

\author{Dennis Jacob\\
UC Berkeley\\
Berkeley, CA, USA\\
{\tt\small djacob18@berkeley.edu}
\and
Chong Xiang\\
Princeton University\\
Princeton, NJ, USA\\
{\tt\small cxiang@princeton.edu}
\and
Prateek Mittal\\
Princeton University\\
Princeton, NJ, USA\\
{\tt\small pmittal@princeton.edu}
}

\begin{document}
\maketitle
\begin{abstract}
Deep learning techniques have enabled vast improvements in computer vision technologies. Nevertheless, these models are vulnerable to adversarial patch attacks which  catastrophically impair performance. The physically realizable nature of these attacks calls for certifiable defenses, which feature provable guarantees on robustness. While certifiable defenses have been successfully applied to single-label classification, limited work has been done for multi-label classification. In this work, we present PatchDEMUX, a certifiably robust framework for multi-label classifiers against adversarial patches. Our approach is a generalizable method which can extend any existing certifiable defense for single-label classification; this is done by considering the multi-label classification task as a series of isolated binary classification problems to provably guarantee robustness. Furthermore, in the scenario where an attacker is limited to a single patch we propose an additional certification procedure that can provide tighter robustness bounds. Using the current state-of-the-art (SOTA) single-label certifiable defense PatchCleanser as a backbone, we find that PatchDEMUX can achieve non-trivial robustness on the MS-COCO and PASCAL VOC datasets while maintaining high clean performance\footnote{Our source code is available at \url{https://github.com/inspire-group/PatchDEMUX}}.
\end{abstract}    
\section{Introduction}
\label{sec:intro}

Deep learning-based computer vision systems have helped transform modern society, contributing to the development of technologies such as self-driving cars, facial recognition, and more \cite{li_survey_2022}. Unfortunately, these performance boosts have come at a security cost; attackers can use \emph{adversarial patches} to perturb patch-shaped regions in images and fool deep learning systems \cite{szegedy_intriguing_2014, brown_adversarial_2018}. The patch threat model presents a unique challenge for the security community due to its physically-realizable nature; for instance, even a single well-designed patch that is printed out can induce failure in the wild \cite{brown_adversarial_2018, eykholt_robust_2018, nesti_evaluating_2021}. 

The importance of adversarial patches has made the design of effective defenses a key research goal. Defense strategies typically fall into one of two categories: \emph{empirical defenses} and \emph{certifiable defenses}. The former leverages clever observations and heuristics to prevent attacks, but can be vulnerable to \emph{adaptive attacks} which bypass the defense through fundamental weaknesses in design \cite{carlini_adversarial_2017, hayes_visible_2018, naseer_local_2018}. As a result, certifiable defenses against patch attacks have become popular for computer vision tasks such as single-label classification and object detection \cite{chiang_certified_2020, levine_randomized_2021, metzen_efficient_2021, xiang_patchguard_2021, xiang_patchguard_2021-1, salman_certified_2022, xiang_objectseeker_2022, xiang_patchcleanser_2022, yatsura_certified_2023, xiang_patchcure_2024}; these methods feature provable guarantees on robustness under any arbitrary patch attack. 

Despite these successes, progress on certifiable defenses against patch attacks has been limited for multi-label classification. Multi-label classifiers provide important capabilities for simultaneously tracking multiple objects while maintaining scalability. Many safety-critical systems depend on the visual sensing capabilities of multi-label classifiers, such as traffic pattern recognition in autonomous vehicles \cite{kungActionSlotVisualActionCentric2024}, video surveillance \cite{elhoseinyMultiClassObjectClassification2013}, and product identification for retail checkout \cite{georgeRecognizingProductsPerexemplar2014}. Some of these applications have become mainstream in industry (i.e., Waymo robotaxis, Just Walk Out checkout, etc.).

To address this challenge we propose PatchDEMUX, a certifiably robust framework against patch attacks for the multi-label classification domain. Our design objective is to extend any existing certifiable defense for single-label classification to the multi-label classification domain. To do so, we leverage the key insight that any multi-label classifier can be separated into individual binary classification tasks. This approach allows us to bootstrap notions of certified robustness based on precision and recall; these are lower bounds on performance which are guaranteed \emph{across all attack strategies in the patch threat model}. We also consider the scenario where an attacker is restricted to a single patch and propose a novel certification procedure that achieves stronger robustness bounds by using constraints in vulnerable patch locations. 

We find that PatchDEMUX achieves non-trivial robustness on the MS-COCO and PASCAL VOC datasets while maintaining high performance on clean data. Specifically, when using the current SOTA single-label certifiable defense PatchCleanser as a backbone, PatchDEMUX attains $85.276\%$ average precision on clean MS-COCO images and $44.902\%$ certified robust average precision. On the PASCAL VOC dataset PatchDEMUX achieves $92.593\%$ clean average precision and $56.030\%$ certified robust average precision. For reference, an undefended model achieves $91.146\%$ average precision on clean MSCOCO images and $96.140\%$ average precision on clean PASCAL VOC images. Overall, the key contributions of our work can be summarized as follows:

\begin{itemize}
    \item We address the challenge of patch attacks in the multi-label domain via a general framework that can interface with any existing/future single-label defense. To the best of our knowledge, our approach is the first of its kind.
    \item Our framework provably guarantees lower bounds on performance irrespective of the chosen patch attack (i.e., the patch can contain an optimized attack, random noise, etc.).
    \item We instantiate a version of our defense framework with the current SOTA single-label defense and achieve strong robust performance on popular benchmarks.
\end{itemize}

\noindent
We hope that future work will integrate with the PatchDEMUX framework and further strengthen the robustness of multi-label classifiers to adversarial patches.
\section{Problem Formulation}
\label{section:background}

In this section, we provide a primer on the multi-label classification task along with standard metrics for evaluation. We next outline the adversarial patch threat model and its relevance in the multi-label setting. Finally, we discuss the concept of certifiable defenses and how they have been used so far to protect single-label classifiers against the patch attack.

\subsection{Multi-label classification}
\label{subsection:multlabel}
Multi-label classification is a computer vision task where images $\mathbf{x} \in \mathcal{X} \subseteq \mathbb{R}^{w \times h \times \gamma}$ with width $w$, height $h$, and number of channels $\gamma$ contain multiple objects simultaneously, with each object belonging to one of $c$ classes \cite{zhang_review_2014}. A classifier is then tasked with recovering each of objects present in an image. Note that this contrasts single-label classification, where exactly one object is recovered from an image.

More rigorously, each input datapoint is a pair $(\mathbf{x}, \mathbf{y})$ where $\mathbf{x} \in \mathcal{X}$ corresponds to an image and $\mathbf{y} \in \mathcal{Y}$ is the associated image label. Each label $\mathbf{y} \in \mathcal{Y} \subseteq \{0, 1\}^c$ is a bitstring where $\mathbf{y}[i] = 1$ means class $i$ is present and $\mathbf{y}[i] = 0$ means class $i$ is absent; this implies that the set of labels $\mathcal{Y}$ is $2^c$ in size, i.e., exponential. A \emph{multi-label classifier} $\mathbb{F}: \mathcal{X}\rightarrow \mathcal{Y}$ is then trained with a loss function such that the predicted label $\mathbf{\hat{y}} := \mathbb{F}(\mathbf{x})$ is equivalent to $\mathbf{y}$. One popular loss function used for training is asymmetric loss (ASL) \cite{ben-baruch_asymmetric_2021}.

To evaluate the performance of a multi-label classifier, it is common to compute the number of \emph{true positives} (i.e., classes $i$ where $\mathbf{y}[i] = \mathbf{\hat{y}}[i] = 1$), the number of \emph{false positives} (i.e., classes $i$ where $\mathbf{y}[i] = 0$ and $\mathbf{\hat{y}}[i] = 1$), and the number of \emph{false negatives} (i.e., classes $i$ where $\mathbf{y}[i] = 1$ and $\mathbf{\hat{y}}[i] = 0$). These can be summarized by the \emph{precision} and \emph{recall} metrics \cite{zhang_review_2014}:
\begin{equation}
\label{eq:precisionrecall}
    precision = \frac{TP}{TP + FP} \quad\quad recall = \frac{TP}{TP + FN}
\end{equation}
The values $TP$, $FP$, and $FN$ represent the number of true positives, false positives, and false negatives respectively.

\subsection{The patch threat model}
\label{subsection:patch}
\textbf{Theoretical formulation.} 
In the patch threat model, attackers possess the ability to arbitrarily adjust pixel values within a restricted region located anywhere on a target image $\mathbf{x} \in \mathcal{X}$; the size of this region can be tuned to alter the strength of the attack \cite{brown_adversarial_2018}. As discussed in \cref{sec:intro}, defending against this threat model is critical due to its physically realizable nature \cite{brown_adversarial_2018, eykholt_robust_2018, nesti_evaluating_2021}. In this paper, we primarily focus on defending against a single adversarial patch as it is a popular setting in prior work \cite{chiang_certified_2020, levine_randomized_2021, metzen_efficient_2021, xiang_patchguard_2021, xiang_patchguard_2021-1, xiang_patchcleanser_2022, xiang_patchcure_2024}. However, our baseline certification methods can also handle multiple patches, provided the underlying single-label defense strategy already has this capability \cite{xiang_patchcleanser_2022}.

We can formally specify patch attacks for an image $\mathbf{x} \in \mathcal{X}$ as follows. Define $\mathcal{R} \subseteq \{0, 1\}^{w \times h}$ as the set of binary matrices which represent restricted regions, where elements inside the region are $0$ and those outside the region are $1$ \cite{xiang_patchcleanser_2022}. Then, the associated patch attacks are:
\begin{equation}
\label{eq:linfty}
    S_{\mathbf{x},\mathcal{R}} := \{\mathbf{r} \circ \mathbf{x} + (\mathbf{1} - \mathbf{r}) \circ \mathbf{x'} | \mathbf{x'} \in \mathcal{X}, \mathbf{r} \in \mathcal{R}\}
\end{equation}
The $\circ$ operator refers to element-wise multiplication with broadcasting to ensure shape compatibility. Note that this formulation demonstrates how the patch attack can be considered a special case of the $\ell_0$-norm threat model \cite{levine_randomized_2021}.

\textbf{Adversarial patches in the multi-label setting.} 
Patch attacks in multi-label classification aim to induce class mismatches between a ground-truth label $\mathbf{y} \in \mathcal{Y}$ and prediction $\mathbf{\hat{y}} \in \mathcal{Y}$. Unlike single-label classification, different types of mismatches are possible in this setting; for instance, patches can increase the number of false negatives and/or the number of false positives predicted by the classifier $\mathbb{F}$. In general, adversarial patches are generated by representing the desired objective as an optimization problem and then applying an iterative technique such as projected gradient descent (PGD) over $S_{\mathbf{x},\mathcal{R}}$ \cite{madry_towards_2019}.

\subsection{Certifiable defenses against patch attacks}
\label{subsection:certifieddefense}
At a high-level, certifiable defenses against patch attacks (CDPA) provide provable guarantees on performance for deep learning-based computer vision systems $\mathbb{F}: \mathcal{X} \rightarrow \mathcal{Y}$ against all possible attacks in the patch threat model \cite{chiang_certified_2020, levine_randomized_2021, metzen_efficient_2021, xiang_patchguard_2021, xiang_patchguard_2021-1, salman_certified_2022, xiang_objectseeker_2022, xiang_patchcleanser_2022, yatsura_certified_2023, xiang_patchcure_2024}. This ensures that defense robustness will not be compromised by future adaptive attacks. 

We formulate a CDPA as having an inference procedure and a certification procedure; additional security parameters, denoted by $\sigma$, manage the trade-off between robust performance and inference time \cite{xiang_patchcleanser_2022}. The inference procedure $INFER_{[\mathbb{F}, \sigma]}: \mathcal{X} \rightarrow \mathcal{Y}$ takes an image $\mathbf{x} \in \mathcal{X}$ as input and outputs a prediction $\hat{\mathbf{y}} \in \mathcal{Y}$. The quality of prediction $\hat{\mathbf{y}}$ with respect to the ground-truth label $\mathbf{y}$ can be evaluated using a performance metric (e.g., precision, recall), which we denote by $\rho:\mathcal{Y}\times\mathcal{Y}\rightarrow\mathbb{R}$. In addition to the inference procedure, the certification procedure $CERT_{[\mathbb{F}, \sigma]}: \mathcal{X} \times\mathcal{Y}\times\mathbb{P}(\mathcal{R})\rightarrow \mathbb{R}$ ($\mathbb{P}()$ denotes power set) takes image $\mathbf{x}$, ground-truth label $\mathbf{y}$, and the threat model represented by the set of allowable patch regions $\mathcal{R}$ to determine the worst possible performance of $INFER$ on image $\mathbf{x}$. The certification procedure is only used for evaluation. Formally, for a performance metric $\rho$ and a patch threat model $S_{\mathbf{x},\mathcal{R}}$ we will have
\begin{equation}
\label{eq:robustnessproperty}
        \rho(INFER_{[\mathbb{F}, \sigma]}(\mathbf{x'}), \mathbf{\bm{y}}) \geq \tau,  \forall \mathbf{x}'\in S_{\mathbf{x},\mathcal{R}}
\end{equation}
Here, $\tau := CERT_{[\mathbb{F}, \sigma]}(\mathbf{x}, \mathbf{y}, \mathcal{R})$ is the lower bound of model prediction quality against an adversary who can use any patch region $\mathbf{r}\in\mathcal{R}$ and introduce arbitrary patch content. Datapoints with a non-trivial lower bound are considered \emph{certifiable}.

We can summarize these concepts as follows.

\begin{definition}[CDPA]
\label{def:certdefensepatches}
    A \emph{certifiable defense against patch attacks (CDPA)} for model $\mathbb{F}: \mathcal{X} \rightarrow \mathcal{Y}$ is a tuple of procedures $DEF := (INFER_{[\mathbb{F}, \sigma]}:\mathcal{X} \rightarrow \mathcal{Y}, CERT_{[\mathbb{F}, \sigma]}:\mathcal{X} \times\mathcal{Y}\times\mathbb{P}(\mathcal{R})\rightarrow \mathbb{R})$ where the former is the inference procedure, the latter is the certification procedure, and $\sigma \subseteq \{0, 1\}^*$ are security parameters. Certifiable datapoints satisfy \cref{eq:robustnessproperty} for a performance metric $\rho: \mathcal{Y} \times \mathcal{Y} \rightarrow \mathbb{R}$.
\end{definition}

We note that we have different $\rho$ for different tasks. For instance, CDPAs for single-label classifiers ensure that the output label is preserved for certifiable datapoints\footnote{We use Iverson bracket notation for convenience}.

\begin{definition}[Single-label CDPA]
\label{def:certdefensesslpatches}
    A \emph{single-label CDPA} is a CDPA for single-label classifiers $\mathbb{F}_s : \mathcal{X} \rightarrow \{1, 2, \dots, c\}$. The performance metric is $\rho(y_1, y_2) := [y_1 = y_2]$. The certification procedure $CERT$ evaluates to $1$ for certifiable datapoints and $0$ otherwise.
\end{definition}

For multi-label classification, we consider the interpretation where the performance metric is $\rho(\mathbf{y_1}, \mathbf{y_2}) := \Sigma_{i = 1}^c [\mathbf{y_1}[i] = 1 \cap \mathbf{y_2}[i] = 1]$ and $CERT$ lower bounds the number of true positives. This helps bootstrap robust metrics such as certified precision and recall (see \cref{subsection:patchdemux_certification}).

\subsection{Certifiable defenses for single-label classifiers against patch attacks}
\label{subsection:certdefensesl}
A variety of CDPA have been developed for single-label classifiers \cite{chiang_certified_2020, levine_randomized_2021, metzen_efficient_2021, xiang_patchguard_2021, xiang_patchguard_2021-1, salman_certified_2022, xiang_patchcleanser_2022, xiang_patchcure_2024}. Current techniques roughly fall into one of two categories: \emph{small receptive field} defenses and \emph{masking} defenses. With regards to the former, the general principle involves limiting the set of image features exposed to the undefended model and then robustly accumulating results across several evaluation calls. Some examples of this approach include De-randomized Smoothing \cite{levine_randomized_2021}, BagCert \cite{metzen_efficient_2021}, and PatchGuard \cite{xiang_patchguard_2021-1, xiang_patchguard_2021}. On the other hand, masking defenses curate a set of masks to provably occlude an adversarial patch regardless of location. PatchCleanser, the current SOTA certifiable defense, uses such a method \cite{xiang_patchcleanser_2022}. Our proposed framework PatchDEMUX is theoretically compatible with any of these techniques.

\section{PatchDEMUX Design}
\label{section:methodology}

\begin{figure*}[!ht]
    \centering
    \includegraphics[width=0.85\linewidth]{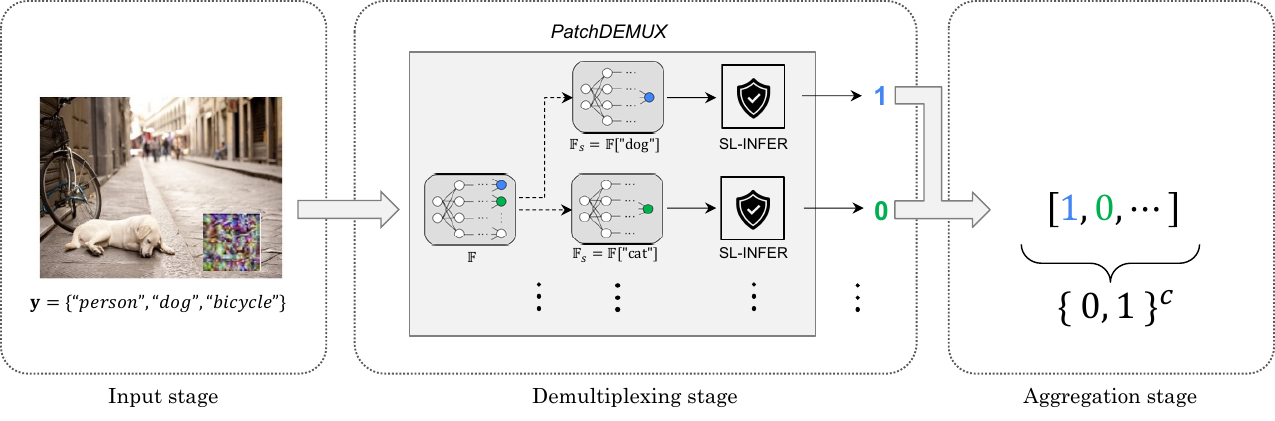}
    \caption{\textit{A diagram which illustrates the defense framework from PatchDEMUX. In the input stage, the (potentially attacked) image is preprocessed. In the demultiplexing stage, the $SL\text{-}INFER$ inference procedure from a single-label CDPA is applied to each individual class in the multi-classification task. This is done by considering the multi-label classifier $\mathbb{F}$ as a series of isolated binary classifiers $\mathbb{F}[i]$ for $i \in \{1, 2, \dots, c\}$. Finally, in the aggregation stage the individual outputs are returned as a single label.}}
    \label{fig:inference_patchdemux}
\end{figure*}

In this section we propose \emph{PatchDEMUX}, a certifiably robust framework for multi-label classifiers against patch attacks. We first outline the key property that any multi-label classification problem can be separated into constituent binary classification tasks. Next, we use this observation to construct a generalizable framework which can theoretically integrate any existing single-label CDPA. We then describe the inference and certification procedures in more detail along with robust evaluation metrics. Finally, we propose a novel location-aware certification method which provides tighter robustness bounds.

\subsection{An overview of the defense framework}
\label{subsection:isolatebinary}
\textbf{Isolating binary classifiers in multi-label classification.} As discussed in \cref{subsection:multlabel}, labels $\mathbf{y} \in \{0, 1\}^c$ in multi-label classification are bitstrings where $\mathbf{y}[i] \in \{0, 1\}$ corresponds to the presence/absence of class $i \in \{1, 2, \dots, c\}$. Note that predictions for each class $\mathbf{y}[i]$ are independent of each other; therefore, the multi-label classification task can be represented as a series of isolated binary classification problems corresponding to each class. This motivates a defense formulation for multi-label classifiers in terms of ``isolated'' binary classifiers, where each class is individually protected by a single-label CDPA. Given a multi-label classifier\footnote{From here on, $\mathcal{Y}$ will denote a multi-label label set with $c$ classes} $\mathbb{F}:\mathcal{X} \rightarrow \mathcal{Y}$, we use the notation $\mathbb{F}[i]: \mathcal{X} \rightarrow \{0, 1\}$ to refer to the isolated classifier for class $i$.

In practice, defining the isolated classifier is complicated as some single-label CDPA designs require architectural restrictions \cite{metzen_efficient_2021, xiang_patchguard_2021-1, xiang_patchguard_2021}. Nevertheless, a workaround is possible; specifically, we can initialize the multi-label classifier $\mathbb{F}: \mathcal{X} \rightarrow \mathcal{Y}$ as an ensemble of $c$ binary classifiers which each satisfy the required architecture. Then, for each class $i \in \{1, 2, \dots, c\}$ we can define the isolated classifier $\mathbb{F}[i]$ as the associated ensemble model. Other defenses are architecture-agnostic \cite{xiang_patchcleanser_2022}. In these cases we can use any off-the-shelf multi-label classifier $\mathbb{F}: \mathcal{X} \rightarrow \mathcal{Y}$ and for each class $i \in \{1, 2, \dots, c\}$ define the isolated classifier $\mathbb{F}[i]$ as having outputs $\mathbb{F}[i](\mathbf{x}) := \mathbb{F}(\mathbf{x})[i]$ for all $\mathbf{x} \in \mathcal{X}$.

\textbf{Our framework.} At a high-level, the PatchDEMUX defense framework takes advantage of the isolation principle to extend any existing single-label CDPA to the multi-label classification task. The \emph{PatchDEMUX inference procedure} consists of three stages (see \cref{fig:inference_patchdemux}). In the input stage, it preprocesses the input image $\mathbf{x} \in \mathcal{X}$. In the demultiplexing stage it isolates binary classifiers for each class $i \in \{1, 2, \dots, c\}$ and applies the underlying single-label CDPA inference procedure. Finally, in the aggregation stage we return the final prediction vector by pooling results from the individual classes. The \emph{PatchDEMUX certification procedure} works similarly. It separately applies the underlying single-label CDPA certification procedure to each isolated classifier and then creates a lower bound for true positives by accumulating the results.

\subsection{PatchDEMUX inference procedure}
\label{subsection:multipatchinference}

\begin{algorithm}
\caption{\textit{The inference procedure associated with PatchDEMUX}}\label{alg:inference_patchdemux}
\begin{algorithmic}[1]

\item[] \textbf{Input:} Image $\mathbf{x} \in \mathcal{X}$, multi-label classifier $\mathbb{F}: \mathcal{X} \rightarrow \mathcal{Y}$, inference procedure $SL\text{-}INFER$ and security parameters $\sigma$ from a single-label CDPA, number of classes $c$
\item[] \textbf{Output:} Prediction $preds \in \{0, 1\}^c$
\Procedure{DemuxInfer}{$\mathbf{x}, \mathbb{F}, SL\text{-}INFER, \sigma, c$}
\State $preds \gets \{0\}^c$ \Comment{Set predictions to zero vector}
\For{$i \gets 1$ to $c$} \Comment{Consider classes individually}
    \State $preds[i] \gets $\Call{SL\text{-}INFER$_{[\mathbb{F}[i], \sigma]}$}{$\mathbf{x}$}
\EndFor
\State \Return{$preds$}
\EndProcedure

\end{algorithmic}
\end{algorithm}

The PatchDEMUX inference procedure is described in \cref{alg:inference_patchdemux}. We first take the inference procedure $SL\text{-}INFER$ from a single-label CDPA and prepare it with security parameters $\sigma$. On line $2$, we initialize a $preds \in \{0, 1\}^c$ array to keep track of individual class predictions. Finally, on line $4$ we run $SL\text{-}INFER$ with the isolated binary classifier $\mathbb{F}[i]$ on input image $\mathbf{x}$ and update $preds$ for class $i$. 

\textbf{Remark.} If the time complexity for $SL\text{-}INFER$ is $\mathcal{O}(f(n))$, the time complexity for \cref{alg:inference_patchdemux} will be $\mathcal{O}(c \cdot f(n))$. However, in practice it is possible to take advantage of relatively negligible defense post-processing and effectively reduce the time complexity to $\mathcal{O}(f(n))$. See \emph{Supplementary Material}, \cref{section:appendix_runtime}.

\subsection{PatchDEMUX certification procedure}
\label{subsection:patchdemux_certification}

\begin{algorithm}
\caption{\textit{The certification procedure associated with PatchDEMUX}}\label{alg:kclasscert}
\begin{algorithmic}[1]
\item[] \textbf{Input:} Image $\mathbf{x} \in \mathcal{X}$, ground-truth $\mathbf{y} \in \mathcal{Y}$, multi-label classifier $\mathbb{F}: \mathcal{X} \rightarrow \mathcal{Y}$, certification procedure $SL\text{-}CERT$ and security parameters $\sigma$ from a single-label CDPA, patch locations $\mathcal{R}$
\item[] \textbf{Output:} Certified number of true positives $TP_{lower}$, false positives upper bound $FP_{upper}$, false negatives upper bound $FN_{upper}$, class certification list $\bm{\kappa}$
\Procedure{DemuxCert}{$\mathbf{x}, \mathbf{y}, \mathbb{F}, SL\text{-}CERT, \sigma, \mathcal{R}$}
  \State $c \gets \text{len}(\mathbf{y})$
  \State $\bm{\kappa} \gets [0]^c$
  \For{$i \gets 1$ to $c$} \Comment{Certify each class separately}
    \State $\bm{\kappa}[i] \gets $\Call{SL\text{-}CERT$_{[\mathbb{F}[i], \sigma]}$}{$\mathbf{x}, \mathbf{y}[i], \mathcal{R}$}
\EndFor
\LineComment{Compute robust metrics}
\State $TP_{lower}, FP_{upper}, FN_{upper} \gets 0, 0, 0$ 
\State $TP_{lower} \gets \Sigma_{i = 1}^c [\bm{\kappa}[i] = 1 \cap \mathbf{y}[i] = 1]$
\State $FP_{upper} \gets \Sigma_{i = 1}^c [\bm{\kappa}[i] = 0 \cap \mathbf{y}[i] = 0]$
\State $FN_{upper} \gets \Sigma_{i = 1}^c [\bm{\kappa}[i] = 0 \cap \mathbf{y}[i] = 1]$
\State \Return $TP_{lower}, FP_{upper}, FN_{upper}, \bm{\kappa}$
\EndProcedure  
\end{algorithmic}
\end{algorithm}
The PatchDEMUX certification procedure is outlined in \cref{alg:kclasscert}. We first initialize the certification procedure $SL\text{-}CERT$ from a single-label CDPA with security parameters $\sigma$. On line $2$, we create the $\bm{\kappa}$ array to store certifiable classes. On line $5$, we run $SL\text{-}CERT$ with the isolated binary classifier $\mathbb{F}[i]$ on datapoint $(\mathbf{x}, \mathbf{y}[i])$ and place the result in $\bm{\kappa}[i]$; recall from \cref{def:certdefensesslpatches} that $SL\text{-}CERT$ returns $1$ for protected datapoints and $0$ otherwise. Finally, on lines $7-10$ we count a successful true positive for classes with $\mathbf{y}[i] = 1$ and $\bm{\kappa}[i] = 1$. Otherwise, we assign a false negative or false positive as we cannot guarantee the accuracy of these classes. We now establish the correctness of these bounds.


\begin{restatable}[\cref{alg:kclasscert} Correctness]{theorem}{patchdemuxcorrectess}
\label{thm:patchdemuxcorrectnessclass}
  Suppose we have an image data point $(\mathbf{x}, \mathbf{y}) \in \mathcal{X} \times \mathcal{Y}$, a single-label CDPA $SL\text{-}DEF$, and a multi-label classification model $\mathbb{F}: \mathcal{X} \rightarrow \mathcal{Y}$. Then, under the patch threat model $S_{\mathbf{x},\mathcal{R}}$ the bounds returned by \cref{alg:kclasscert} are correct.
\end{restatable}

\begin{proof}
See \emph{Supplementary Material}, \cref{section:appendix_certproofs}.
\end{proof}

Thus, using \cref{alg:kclasscert} we have established the lower bound on true positives ($TP_{lower}$) and the upper bound on both false positives ($FP_{upper}$) and false negatives ($FN_{upper}$) when using \cref{alg:inference_patchdemux}. This allows us to bootstrap notions of \emph{certified precision} and \emph{certified recall} by referencing \cref{eq:precisionrecall}:

\begin{equation}
\label{eq:certifiedprecision}
    certified \: precision = \frac{TP_{lower}}{TP_{lower} + FP_{upper}}
\end{equation}
\begin{equation}
\label{eq:certifiedrecall}
    certified \: recall = \frac{TP_{lower}}{TP_{lower} + FN_{upper}}
\end{equation}

\noindent
Note by construction that both metrics provide lower bounds for precision and recall on a datapoint $(\mathbf{x}, \mathbf{y})$ \emph{irrespective of any attempted patch attack}; the real-world performance of our defense will always be higher. Therefore, an empirical evaluation of existing multi-label attack vectors is not necessary \cite{aichGAMAGenerativeAdversarial2022, melacciDomainKnowledgeAlleviates2022, aichLeveragingLocalPatch2023, mahmoodSemanticAwareMultiLabelAdversarial2024}. Furthermore, micro-averaging these metrics across datapoints provides lower bounds on precision and recall for an entire dataset \cite{zhang_review_2014}.

\subsection{Location-aware certification}
\label{subsection:newcertmethod}
\begin{figure*}[!ht]
    \centering
    \includegraphics[width=\linewidth]{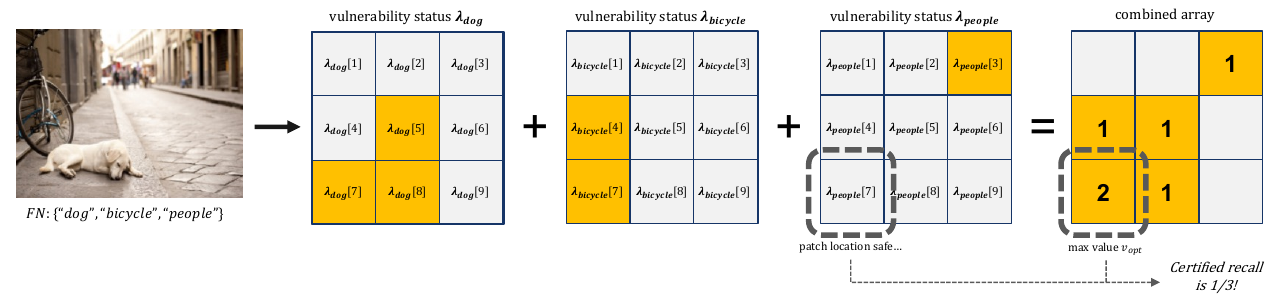}
    \caption{\textit{A diagram which illustrates the key intuition for the location-aware approach. In the sample image we assume all three objects (i.e., ``dog'', ``bicycle'', ``people'') are false negatives. Thus, for each $FN$ we extract the vulnerability status over all patch locations (orange means vulnerable) and accumulate them to find the most vulnerable patch location; this happens to be in the bottom left corner of the image. However, the ``people'' class by itself is not vulnerable to this location; thus, we can claim stronger robustness bounds than initially suggested by \cref{alg:kclasscert}.}}
    \label{fig:locationaware_ex}
\end{figure*}

We now discuss an improved method called \emph{location-aware certification} which extends \cref{alg:kclasscert}. This method works in the scenario where an attacker is restricted to a single patch. The general intuition is that if we track vulnerable patch locations for each class, we can use the constraint that an adversarial patch can only be placed at one location to extract stronger robustness guarantees. For instance, suppose we have an image with a dog, a bicycle, and people (see \cref{fig:locationaware_ex}). If we directly apply \cref{alg:kclasscert}, it is possible that each of these classes would individually fail to be certified. However, this method does not account for the fact that different classes may be vulnerable at different locations; for example, the ``dog'' and ``bicycle'' classes might be at risk in the bottom left corner of the image, while the ``people'' class is at risk near the top. Because the patch cannot exist in two places simultaneously, at least one class must be robust and the actual certified recall will be $1/3$.

\subsubsection{Tracking vulnerable patch locations}
\label{subsubsection:singlelabellocal}
We now give a formal treatment of our core idea. Suppose we have a single-label CDPA $SL\text{-}DEF$. For many existing single-label defenses, it is possible to relate the certification procedure $SL\text{-}CERT$ to the complete list of patch locations $\mathcal{R}$ from \cref{eq:linfty} \cite{chiang_certified_2020, levine_randomized_2021, metzen_efficient_2021, xiang_patchguard_2021, xiang_patchguard_2021-1, salman_certified_2022, xiang_patchcleanser_2022, xiang_patchcure_2024}. In these cases, we extend \cref{def:certdefensesslpatches} and allow $SL\text{-}CERT$ to return a \emph{vulnerability status array}, which we denote by $\bm{\lambda} \in \{0, 1\}^{|\mathcal{R}|}$. A value of $1$ implies the image $\mathbf{x} \in \mathcal{X}$ is protected from attacks located in $\mathbf{r} \in \mathcal{R}$, while $0$ means it is not. 

This provides a convenient formulation with which to express our improved method. Consider a multi-label classifier $\mathbb{F}: \mathcal{X} \rightarrow \mathcal{Y}$. We first obtain vulnerability status arrays $\bm{\lambda}$ for each class in \cref{alg:kclasscert} that could not be certified; this is done by isolating the associated binary classifiers. We then note that given $k$ classes of a common failure mode (i.e., $FN$ or $FP$), the sum of the inverted arrays $\mathbf{1} - \bm{\lambda}$ will represent the frequency of the failure type at each patch location. The key insight is that the maximum value, $v_{opt}$, from the combined array will represent the patch location $\mathbf{r_{opt}} \in \mathcal{R}$ of the image most vulnerable to a patch attack; an attacker must place an adversarial patch at this location to maximize malicious effects. Note however that it is possible $v_{opt} < k$. Then, as per the construction of each $\bm{\lambda}$ these $k - v_{opt} > 0$ classes will be \emph{guaranteed robustness under the optimal patch location.}

\subsubsection{Proposing our novel algorithm}
\label{subsubsection:newcertproposal}

\begin{algorithm}
\caption{\textit{Location-aware certification for $FN$}}\label{alg:newcert}
\begin{algorithmic}[1]
\item[] \textbf{Input:} Image $\mathbf{x} \in \mathcal{X}$, ground-truth $\mathbf{y} \in \mathcal{Y}$, multi-label classifier $\mathbb{F}: \mathcal{X} \rightarrow \mathcal{Y}$, certification procedure $SL\text{-}CERT$ and security parameters $\sigma$ from a single-label CDPA, patch locations $\mathcal{R}$
\item[] \textbf{Output:} Certified number of true positives $TP_{new}$, false negatives upper bound $FN_{new}$
\Procedure{LocCert}{$\mathbf{x}, \mathbf{y}, \mathbb{F}, SL\text{-}CERT, \sigma, \mathcal{R}$}
\LineComment{Pass all args to \Call{DemuxCert}{$\dots$}}
\State $TP, FP, FN, \bm{\kappa} \gets $ \Call{DemuxCert}{$\dots$}
\LineComment{Initialize array with list of $FN$ indices}
\State $c \gets \text{len}(\mathbf{y})$
\State $fnIdx \gets \text{list}(\{1 \leq i \leq c : \bm{\kappa}[i] = 0 \cap \mathbf{y}[i] = 1\})$ 
\State $fnCertFails \gets [0]^{FN \times |\mathcal{R}|}$
\For{$k \gets 1$ to $FN$} \Comment{Isolate each $FN$ classifier}
    \State $\mathbb{F}_s \gets \mathbb{F}[fnIdx[k]]$ 
    \State $\bm{\lambda} \gets$ \Call{SL\text{-}CERT$_{[\mathbb{F}_s, \sigma]}$}{$\mathbf{x}, \mathbf{y}[fnIdx[k]], \mathcal{R}$}
    \State $fnCertFails[k] = \bm{1} - \bm{\lambda}$
\EndFor

\State $fnTotal \gets \text{sum}(fnCertFails, \text{dim} = 0)$
\State $FN_{new} = \text{max}(fnTotal)$ \Comment{Pick worst location}
\State $TP_{new} = TP + (FN - FN_{new})$
\State \Return $TP_{new}, FN_{new}$
\EndProcedure  
\end{algorithmic}
\end{algorithm}

These insights are encapsulated by \cref{alg:newcert}, the location-aware certification method for false negatives.\footnote{Obtaining $FP_{new}$ is similar, with line $5$ changed to track $FP$ indices} It works by first computing robustness bounds for data point $(\mathbf{x}, \mathbf{y})$ via \cref{alg:kclasscert}. On line $5$ we determine the false negative classes that failed certification in \cref{alg:kclasscert}. During the $for$ loop on lines $8-12$, we extract the vulnerability status array $\bm{\lambda}$ for each false negative by isolating the associated binary classification task. Finally, we sum the inverted arrays $\mathbf{1} - \bm{\lambda}$ on line $13$ and pick the patch location with the largest value; this is the max number of false negatives an attacker can induce at test time. We then alter the lower bound for true positives on line $16$.

We now demonstrate that \cref{alg:newcert} provides superior bounds to \cref{alg:kclasscert}.

\begin{restatable}[\cref{alg:newcert} Correctness]{theorem}{histcorrect}
\label{thm:histcorrect}
    Suppose we have an image data point $(\mathbf{x}, \mathbf{y}) \in \mathcal{X} \times \mathcal{Y}$, a single-label CDPA $SL\text{-}DEF$, and a multi-label classification model $\mathbb{F}: \mathcal{X} \rightarrow \mathcal{Y}$. If $SL\text{-}CERT$ returns the vulnerability status array $\bm{\lambda}$ associated with each $\mathbf{r} \in \mathcal{R}$, then under the patch threat model $S_{\mathbf{x},\mathcal{R}}$ the bounds from \cref{alg:newcert} are correct and stronger than \cref{alg:kclasscert}.
\end{restatable}

\begin{proof}
See \emph{Supplementary Material}, \cref{section:appendix_certproofs}.
\end{proof}

An analogue to \cref{thm:histcorrect} also exists for $FP$ bounds, and can be proved using a modified version of \cref{alg:newcert} that tracks $FP$ indices.
\section{Main Results}
\label{section:mainresults}
\subsection{Setup}
\label{subsection:methodology}

In this section, we discuss our evaluation setup. The associated source code is available at \url{https://github.com/inspire-group/PatchDEMUX}.

\textbf{Backbone initialization and parameters.} Recall from \cref{subsection:isolatebinary} that PatchDEMUX requires an underlying single-label CDPA to operate. For our experiments we choose PatchCleanser, as it is the current SOTA single-label CDPA and is architecture-agnostic (i.e., it is compatible with any off-the-shelf multi-label classifier) \cite{xiang_patchcleanser_2022}. PatchCleanser works by using a novel double-masking algorithm along with a specially generated certification mask set to provably remove adversarial patches \cite{xiang_patchcleanser_2022}. The mask generation process has two security parameters. The first is the number of masks for each image dimension $k_1 \times k_2$; using more masks leads to longer inference time but results in stronger robustness, effectively serving as a ``computational budget'' \cite{xiang_patchcleanser_2022}. The second is the estimated size of the patch $p$ in pixels. Our experiments with PatchDEMUX use $6 \times 6$ masks and assume the patch is ${\sim}2\%$ of the overall image size, which are the default settings in \citet{xiang_patchcleanser_2022}; we vary these parameters in \cref{section:securityparam}. For more details on how PatchCleanser fits into the PatchDEMUX framework see \emph{Supplementary Material}, \cref{section:appendix_patchcleanser}. 

We note that PatchCleanser can also provide protection against multiple patches \cite{xiang_patchcleanser_2022}. Because our baseline certification method provably extends single-label guarantees to multi-label setting, it will also feature resistance against multiple patches. In our experiments, we focus on the single patch setting for simplicity.

\textbf{Dataset and model architectures.} We evaluate our defense on two datasets: MS-COCO \cite{lin_microsoft_2015} and PASCAL VOC \cite{everinghamPascalVisualObject2010}. The former is a challenging collection of images that feature ``common objects in context'' \cite{lin_microsoft_2015}, while the latter focuses on ``realistic scenes'' \cite{everinghamPascalVisualObject2010}. For our experiments we test on the MS-COCO 2014 validation split, which contains ${\sim}41,000$ images and $80$ classes, and the PASCAL VOC 2007 test split, which has ${\sim}5,000$ images and $20$ classes. Both of these splits are commonly used in the multi-label classification community \cite{ben-baruch_asymmetric_2021, liu_query2label_2021, ridnik_ml-decoder_2023, xu_open_2023}.

For the multi-label classifier architecture, we evaluate two options. The first is a ResNet-based architecture from  \citet{ben-baruch_asymmetric_2021} that uses convolution kernels and has an input size of $448 \times 448$. The second is a vision transformer-based (ViT) architecture from \citet{liu_query2label_2021} that uses the self-attention mechanism and has an input size of $384 \times 384$ \cite{dosovitskiy_image_2021, liu_query2label_2021, xu_open_2023}. These models are chosen as they perform well on the multi-label classification task and have publicly available checkpoints. We resize images to fit on each model and apply different defense fine-tuning methods (i.e., Random Cutout \cite{devries_improved_2017}, Greedy Cutout \cite{saha_revisiting_2023}) to achieve stronger robustness guarantees.

\textbf{Evaluation settings and metrics.} Our results feature several evaluation settings.
\begin{enumerate}
    \item \emph{Undefended clean:} This setting represents evaluation on clean data without the PatchDEMUX defense.
    \item \emph{Defended clean:} This setting refers to evaluation on clean data with the PatchDEMUX defense activated.
    \item \emph{Certified robust:} This setting represents lower bounds on performance determined using \cref{alg:kclasscert}.
    \item \emph{Location-aware robust:} This setting represents the tighter certification bounds from \cref{alg:newcert}. We report performance corresponding to the worst-case attacker (see \emph{Supplementary Material}, \cref{section:appendix_locablation}).
\end{enumerate}
The first two are \emph{clean settings}, where precision and recall metrics are empirically computed for each datapoint. The latter two are \emph{certified robust settings}, where certified precision and certified recall metrics are computed using \cref{alg:kclasscert} and \cref{alg:newcert}. In all four evaluation settings we micro-average metrics over the entire dataset \cite{zhang_review_2014}. In addition, we sweep model outputs across a range of threshold values to create \emph{precision-recall plots}. The associated area-under-curve values aggregate performance and are used to approximate \emph{average precision} (AP); more details are in \emph{Supplementary Material}, \cref{section:appendix_evalmetrics}.

\subsection{PatchDEMUX overall performance}
\label{subsection:bestresults}
\begin{table*}[!ht]
    \centering
    \caption{\textit{PatchDEMUX performance with ViT architecture on the MS-COCO 2014 validation dataset. Precision values are evaluated at key recall levels along with the approximated average precision. We assume the patch attack is at most $2\%$ of the image area and use a computational budget of $6 \times 6$ masks}}
    \subcaptionbox{\textit{Clean setting precision values}\label{tab:vitbesttable_clean}}{
    \begin{tabular}{lcccc}
        \toprule
        \textbf{Architecture} & \multicolumn{4}{c}{ViT} \\
        \cmidrule(l){1-1}\cmidrule(l){2-5}
        \textbf{Clean recall} & $25\%$ & $50\%$ & $75\%$ & $AP$ \\
        \midrule
        \textit{Undefended} & 99.930 & 99.704 & 96.141 & 91.146 \\
        \textit{Defended} & 99.894 & 99.223 & 87.764 & 85.276 \\
        \bottomrule
    \end{tabular}
    }
    \subcaptionbox{\textit{Certified robust setting precision values}\label{tab:vitbesttable_robust}}{
    \begin{tabular}{lcccc}
        \toprule
        \textbf{Architecture} & \multicolumn{4}{c}{ViT} \\
        \cmidrule(l){1-1}\cmidrule(l){2-5}
        \textbf{Certified recall} & $25\%$ & $50\%$ & $75\%$ & $AP$ \\
        \midrule
        \textit{Certified robust} & 95.369 & 50.950 & 22.662 & 41.763 \\
        \textit{Location-aware} & 95.670 & 56.038 & 26.375 & 44.902\\
        \bottomrule
    \end{tabular}
    }
    \label{tab:bestdata}
\end{table*}

\begin{figure*}[!ht]
    \centering
    \subcaptionbox{\textit{Clean setting precision-recall curves}\label{fig:vitbestprecrecall_clean}}{\includegraphics[width=0.4\textwidth]{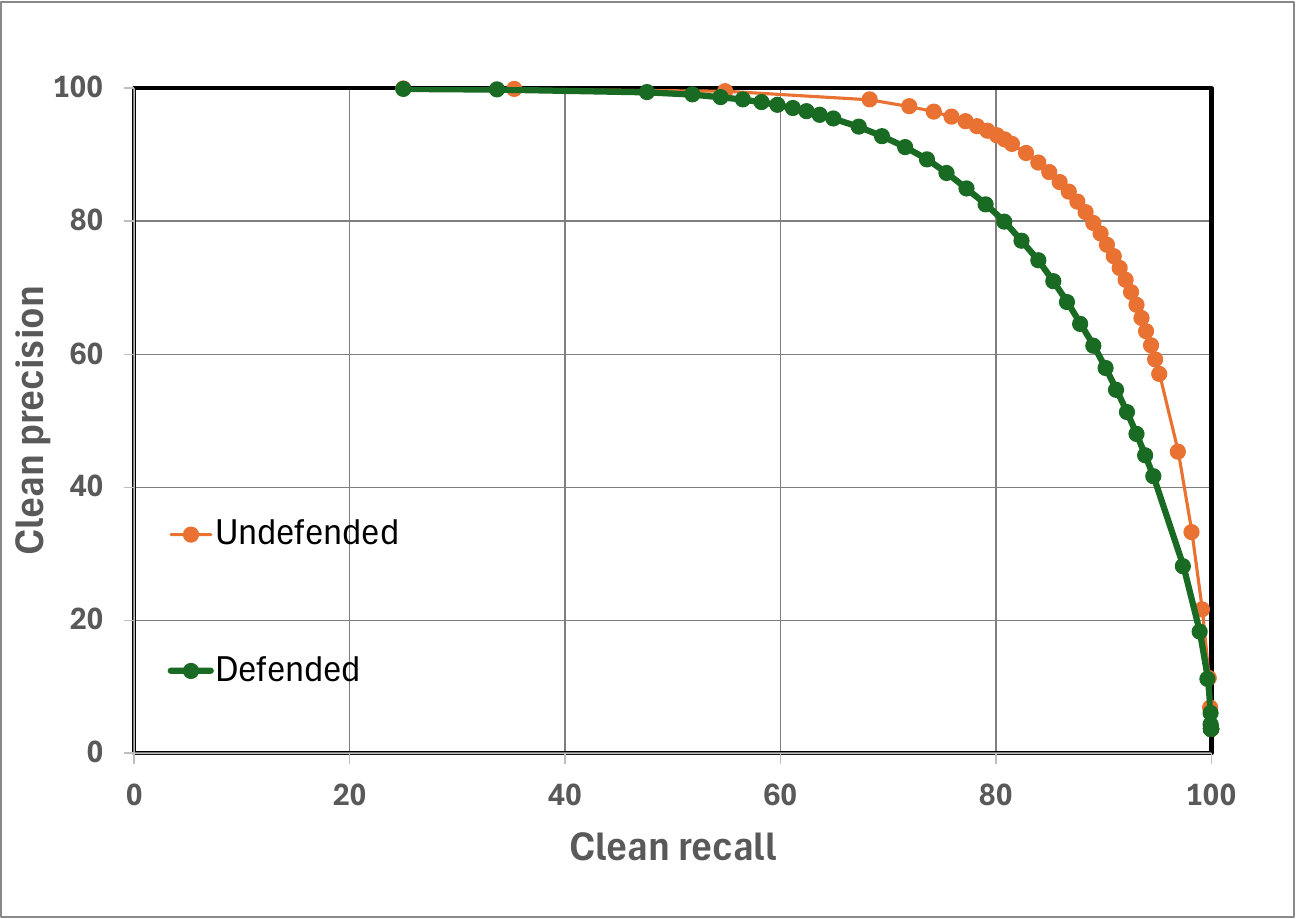}}
    \subcaptionbox{\textit{Certified robust setting precision-recall curves}\label{fig:vitbestprecrecall_robust}}{\includegraphics[width=0.4\textwidth]{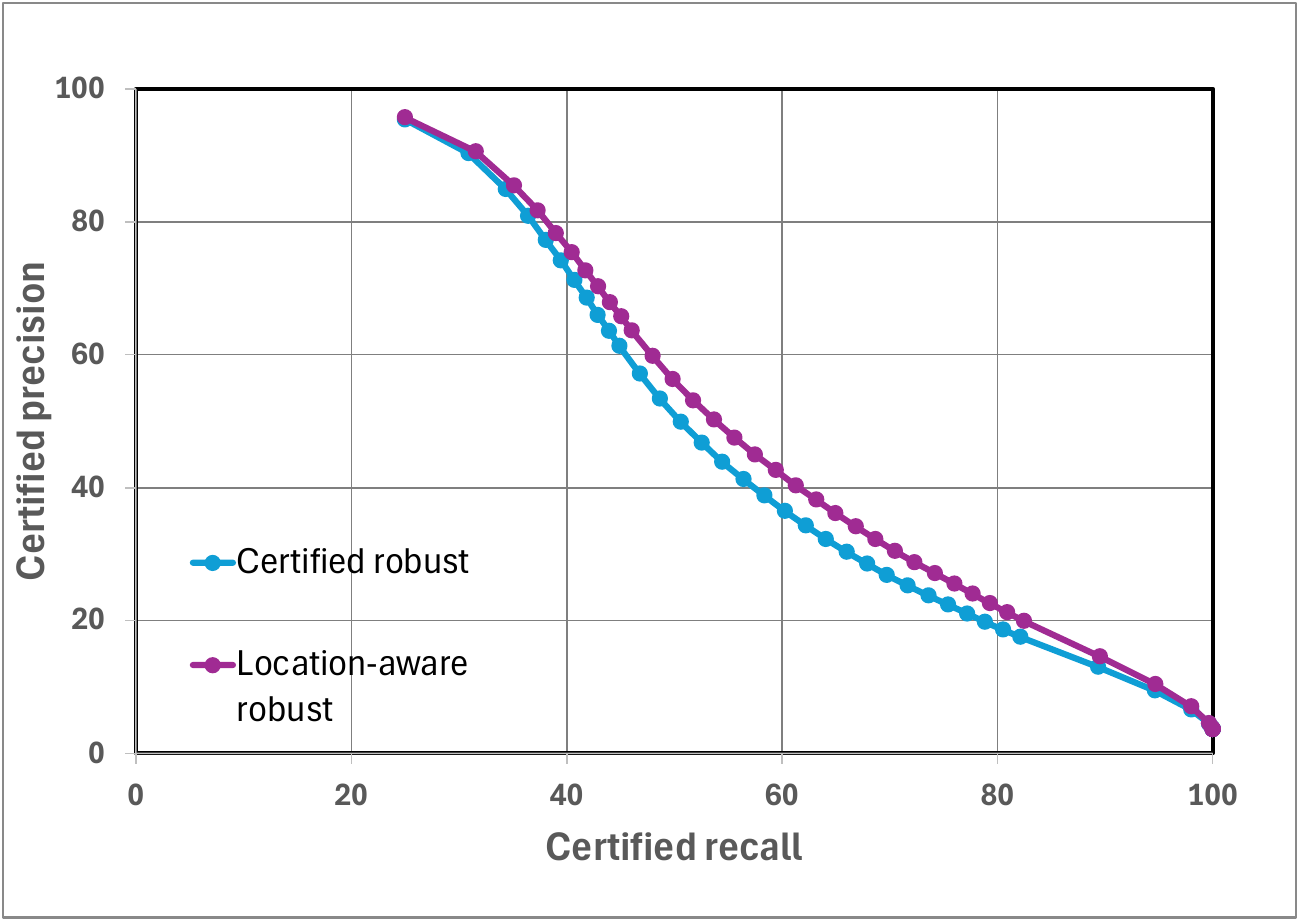}}
    \caption{\textit{PatchDEMUX precision-recall curves with ViT architecture over the MS-COCO 2014 validation dataset. We consider the clean and certified robust evaluation settings. We assume the patch attack is at most $2\%$ of the image area and use a computational budget of $6 \times 6$ masks.}}
    \label{fig:vitbestprecrecall}
\end{figure*}

In this section we report our main results for PatchDEMUX on the MS-COCO 2014 validation dataset. We summarize the precision values associated with key recall levels in \cref{tab:bestdata}. \cref{fig:vitbestprecrecall} features precision-recall plots, while AP values are present in \cref{tab:bestdata}. Because the ViT architecture outperforms the Resnet architecture (see \emph{Supplementary Material}, \cref{section:appendix_resnet}) we focus on the ViT model here. Performance of the ViT architecture on the PASCAL VOC 2007 test dataset is in \emph{Supplementary Material}, \cref{section:appendix_pascalvoc}.

\textbf{High clean performance.}
As shown in \cref{tab:vitbesttable_clean} and \cref{fig:vitbestprecrecall_clean}, the PatchDEMUX inference procedure features excellent performance on clean data. Specifically, the defended clean setting achieves ${\sim}94\%$ of the undefended model's AP. These results demonstrate that PatchDEMUX can be deployed at test time with minimal loss in performance utility. 

\textbf{Non-trivial robustness.} 
\cref{tab:vitbesttable_robust} and \cref{fig:vitbestprecrecall_robust} also show that PatchDEMUX attains non-trivial certifiable robustness on the MS-COCO 2014 validation dataset. For instance, when fixed at $50\%$ certified recall PatchDEMUX achieves $56.038\%$ certified precision. This performance remains stable across a variety of thresholds, as evidenced by the $44.902\%$ certified AP value. Location-aware certification is a key factor in these results, improving certified AP by almost $3$ points compared to the certified robust setting. Improvements are most notable in the \emph{mid recall-mid precision} region of the certified robust precision-recall plot (\cref{fig:vitbestprecrecall_robust}).

Interestingly, the defended clean precision-recall plot (\cref{fig:vitbestprecrecall_clean}) is concave in shape while the certified robust plots (\cref{fig:vitbestprecrecall_robust}) are slightly convex. This performance gap is likely due to the sensitivity of PatchCleanser's certification procedure to object occlusion from the generated mask set. This limitation is compounded by the fact that many MS-COCO images contain objects that are small relative to the overall image size \cite{lin_microsoft_2015, xiang_patchcleanser_2022}.

\subsection{Ablation studies}
\label{subsection:ablations}
We also perform a series of ablation studies for PatchDEMUX using the MS-COCO 2014 validation dataset. We first empirically compare different attackers in the location-aware robust setting and find that attacks targeting false positives are relatively ``weaker'' (see \emph{Supplementary Material}, \cref{section:appendix_locablation}). We then investigate the impact of different defense fine-tuning routines, and find that variants of cutout fine-tuning (i.e., Random Cutout \cite{devries_improved_2017}, Greedy Cutout \cite{saha_revisiting_2023}) can boost model robustness (see \emph{Supplementary Material}, \cref{section:appendix_pretrainablation}); the strongest results for the defended clean setting are featured in the previous section.
\section{Security Parameter Experiments}
\label{section:securityparam}

\begin{figure*}[!ht]
    \centering
    \subcaptionbox{\textit{Clean AP and certified AP as a function of mask number.}\label{fig:secparameteranalysis_masknum}}{\includegraphics[width=0.425\textwidth]{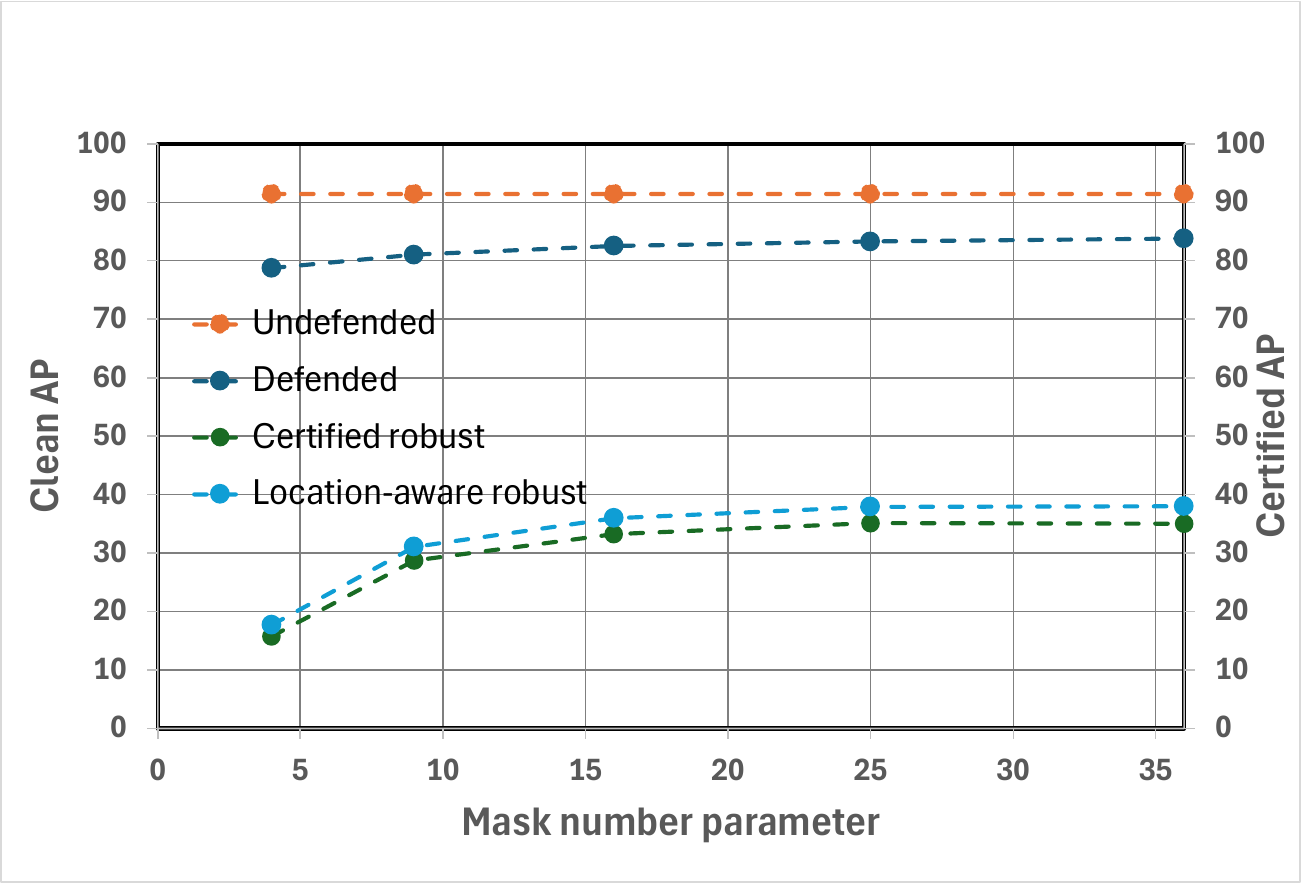}}
    \subcaptionbox{\textit{Clean AP and certified AP as a function of patch size.}\label{fig:secparameteranalysis_patch}}{\includegraphics[width=0.425\textwidth]{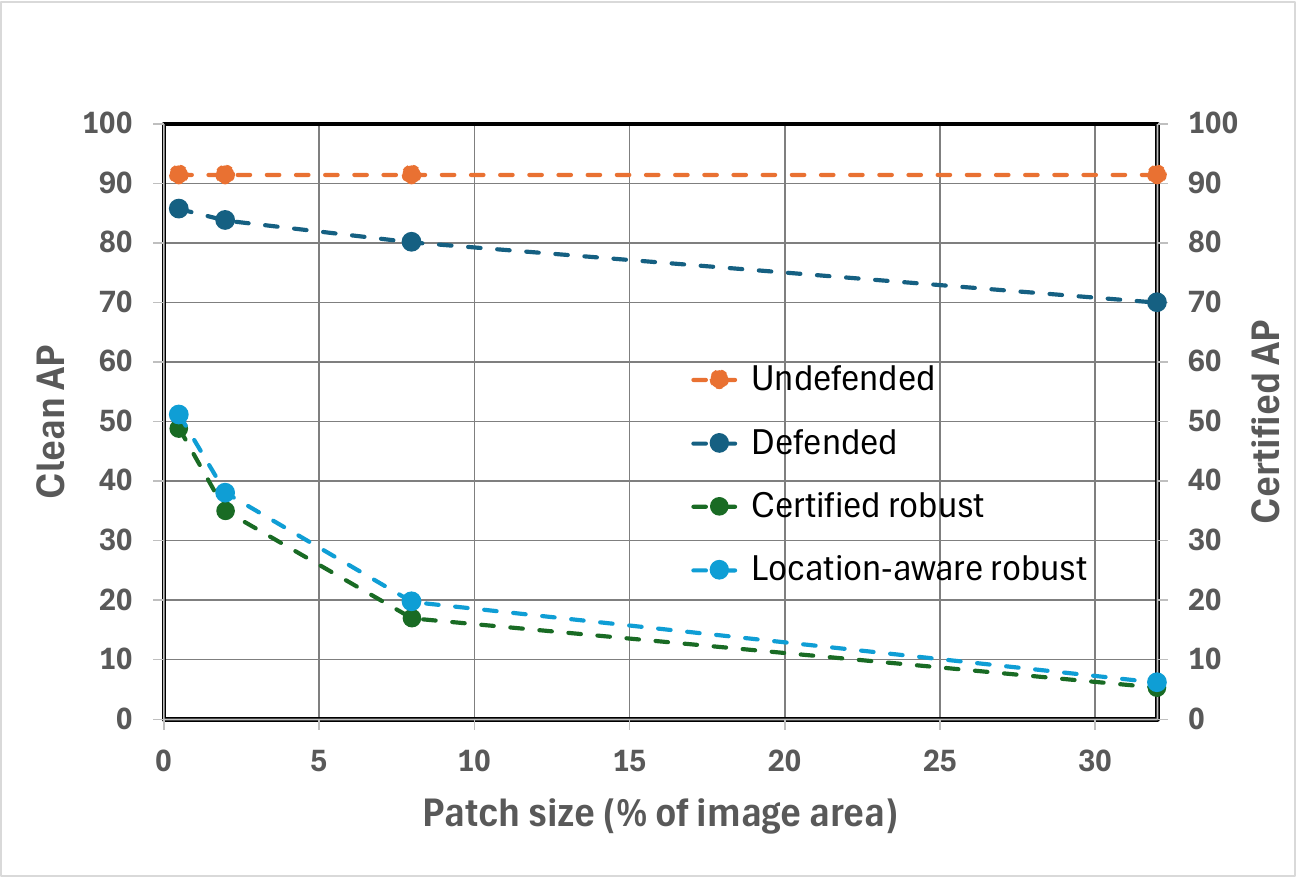}}
    \caption{\textit{The impact of varying PatchCleanser security parameters on PatchDEMUX performance. Experiments performed on MS-COCO 2014 validation dataset. We compute clean AP for the clean setting evaluations, and certified AP for the certified robust setting evaluations.}}
    \label{fig:secparameteranalysis}
\end{figure*}

As discussed in \cref{subsection:methodology}, the PatchCleanser backbone has two security parameters: the number of masks desired in each dimension $k_1 \times k_2$ (i.e., the ``computational budget'') and the estimated size of the patch $p$ in pixels \cite{xiang_patchcleanser_2022}. In this section, we study the impact of these parameters on PatchDEMUX performance. To isolate the effects of security parameter variation, we use ViT checkpoints without defense fine-tuning. Experiments are done on the MS-COCO 2014 validation dataset.

\subsection{Impact of varying mask number}
\label{subsection:masknumexp}
We present results when varying the mask number parameter in \cref{fig:secparameteranalysis_masknum} (the associated table is in \emph{Supplementary Material}, \cref{section:appendix_sectables}). We assume the number of masks in each dimension is the same (i.e., $k := k_1 = k_2$) and evaluate with respect to $k^2$. We keep the patch size parameter its default value of ${\sim}2\%$.

\textbf{Limited tradeoff between computational budget and robustness.} We find that PatchDEMUX provides consistent defended clean and certified robust performance even after greatly reducing the number of masks. For instance, decreasing the number of masks from $36$ to $16$ results in a maximum AP drop of $2$ points across all evaluation settings. At the extreme of $k^2 = 4$ masks more substantial performance drops are noticeable. This is expected, as the mask generation method from PatchCleanser will create larger masks to compensate for reduced mask number; this leads to increased occlusion and fewer certification successes \cite{xiang_patchcleanser_2022}.

\subsection{Impact of varying patch size}
\label{subsection:patchsizeexp}
We present results when varying the patch size estimate in \cref{fig:secparameteranalysis_patch} (the associated table is in \emph{Supplementary Material}, \cref{section:appendix_sectables}). We keep the mask number at its default value of $6 \times 6$ masks.

\textbf{Strong clean performance over different patch sizes.} We find that the defended clean performance of PatchDEMUX is resilient to increasing patch size; indeed, clean AP only drops from $85.731$ in the smallest patch setting to $69.952$ in the largest. Thus, even in unlikely scenarios (i.e., a patch size of $\geq32\%$ would be easily detectable by hand) PatchDEMUX maintains strong inference performance. For the certified robust settings, PatchDEMUX provides relatively strong robustness guarantees on smaller patches (i.e., $\leq2\%$) and performance degrades for larger patches (i.e., $\geq8\%$); certified AP drops close to $0\%$ when a patch size of $32\%$ is considered. These trends align with experiments done by \citet{xiang_patchcleanser_2022} in the single-label classification domain; the general intuition is that larger patch sizes require PatchCleanser to generate larger masks, making certification failures more likely \cite{xiang_patchcleanser_2022}.

\subsection{Overall takeaways}
\label{subsection:takeaways}
Overall, we find that PatchDEMUX performance tradeoffs corroborate with findings from \citet{xiang_patchcleanser_2022}. This illustrates a key feature of our defense framework: PatchDEMUX successfully adapts the strengths of underlying single-label CDPAs to the multi-label classification setting.
\section{Related Work}
\label{section:relatedwork}
\textbf{Certifiable defenses against patch attacks.} CDPAs have been designed for various computer vision applications. In single-label classification, defense strategies include bound propagation methods \cite{chiang_certified_2020}, small receptive field methods \cite{levine_randomized_2021, metzen_efficient_2021, xiang_patchguard_2021-1, xiang_patchguard_2021}, and masking methods \cite{xiang_patchcleanser_2022}. CDPAs have also been proposed for object detection \cite{xiang_objectseeker_2022} and semantic segmentation \cite{yatsura_certified_2023}, although notions of certifiable robustness are more difficult to define in these domains.

\textbf{Certifiable defenses in multi-label classification.} \citet{jia_multiguard_2022} proposed MultiGuard, a certifiably robust defense for multi-label classifiers that generalizes randomized smoothing \cite{cohen_certified_2019}. However, MultiGuard is designed to protect against $\ell_2$-norm attacks and does not address adversarial patches.

\section{Conclusion}
\label{section:conclusion}
The threat of adversarial patch attacks has compromised real-world computer vision systems, including those that depend on multi-label classifiers. To this end we introduced PatchDEMUX, a certifiably robust framework for multi-label classifiers against adversarial patches. PatchDEMUX can extend any existing single-label CDPA, including the current SOTA single-label CDPA PatchCleanser, and demonstrates strong performance on the MS-COCO and PASCAL VOC datasets. We hope that future work will take advantage of our modular framework to significantly mitigate the impact of adversarial patches. 
\section{Acknowledgements}
\label{section:Acknowledgements}
We would like to thank the anonymous CVPR reviewers for their helpful feedback. This work was supported by National Science Foundation grants IIS-2229876 (the ACTION center) and CNS-2154873. Prateek Mittal acknowledges the support of NSF grant CNS-2131938, Princeton SEAS Innovation award, and OpenAI \& FarAI superalignment grants.
{
    \small
    \bibliographystyle{ieeenat_fullname}
    \bibliography{author-kit-CVPR2025-v3.1-latex-/sources}
}

\clearpage

\setcounter{page}{1}
\maketitlesupplementary
\appendix

\section{Certification Robustness Proofs}
\label{section:appendix_certproofs}

\subsection{Baseline certification correctness}
\label{subsection:baselinecertcorrect}

In this section, we provably demonstrate robustness for our baseline certification procedure. Specifically, we prove \cref{thm:patchdemuxcorrectnessclass}, which ensures correctness of the bounds returned by \cref{alg:kclasscert}. For convenience, we re-state the theorem.

\patchdemuxcorrectess*

\begin{proof}
We first demonstrate that classes included in $TP_{lower}$ will be guaranteed correctness. Consider an arbitrary class $i^* \in \{1, 2, \dots, c\}$ with label $\mathbf{y}[i^*] = 1$. If this class is included in $TP_{lower}$, then we must have $\bm{\kappa}[i^*] = 1$ (i.e., line $9$ in \cref{alg:kclasscert}). This implies that on line $5$ we must have $SL\text{-}CERT_{[\mathbb{F}[i^*], \sigma]}(\mathbf{x}, \mathbf{y}[i^*], \mathcal{R}) = 1$. Now consider when \cref{alg:inference_patchdemux} reaches index $i^* \in \{1, 2, \dots, c\}$ in the \emph{for} loop on line $3$. Because the datapoint $(\mathbf{x}, \mathbf{y}[i^*])$ was certifiable, by \cref{def:certdefensesslpatches} we will have
\begin{equation*}
  SL\text{-}INFER_{[\mathbb{F}[i^*], \sigma]}(\mathbf{x'}) = 1 \quad \forall \mathbf{x}' \in S_{\mathbf{x}, \mathcal{R}}
\end{equation*}
This implies that every class accounted for in $TP_{lower}$ will be successfully recovered by \cref{alg:inference_patchdemux} regardless of the attempted patch attack.

Next, we demonstrate that classes included in $FN_{upper}$ will not be guaranteed correctness. Consider an arbitrary class $i^* \in \{1, 2, \dots, c\}$ with label $\mathbf{y}[i^*] = 1$. In this case we will have $\bm{\kappa}[i^*] = 0$, and thus classes included in $FN_{upper}$ will have $SL\text{-}CERT_{[\mathbb{F}[i^*], \sigma]}(\mathbf{x}, \mathbf{y}[i^*], \mathcal{R}) = 0$. Now consider when \cref{alg:inference_patchdemux} reaches index $i^* \in \{1, 2, \dots, c\}$ in the \emph{for} loop on line $3$. By \cref{def:certdefensesslpatches} it is possible that 
\begin{equation*}
  \exists \mathbf{x}' \in S_{\mathbf{x}, \mathcal{R}} \quad | \quad SL\text{-}INFER_{[\mathbb{F}[i^*], \sigma]}(\mathbf{x'}) = 0
\end{equation*}
Essentially, in the worst-case scenario these classes might be mispredicted and be false negatives. Thus, none of the classes included in $FN_{upper}$ can be guaranteed correctness. Because every class with $\mathbf{y}[i^*] = 1$ will be accounted for by either $TP_{lower}$ or $FN_{upper}$ (mutually exclusive), we conclude that $TP_{lower}$ will be the correct lower bound for objects recovered and $FN_{upper}$ will be the correct upper bound for objects missed.

The correctness of the $FP_{upper}$ bound can be shown in a similar fashion, albeit by considering classes with $\mathbf{y}[i^*] = 0$.

\end{proof}

\subsection{Location-aware certification correctness}
\label{subsection:loccertcorrect}

In this section we demonstrate the correctness of our novel location-based certification method. To do so, it is helpful to use the following lemma.

\begin{restatable}[\cref{alg:newcert} Tightness]{lemma}{histtight}
\label{lemma:histtight}
    Given that we have derived a bound on $FN$ using the technique from \cref{alg:kclasscert}, \cref{alg:newcert} will return a new bound $FN_{new} \leq FN$.
\end{restatable}
\begin{proof}
We will show that $FN_{new}$ provides a tighter bound (i.e., the inequality $FN_{new} \leq FN$ is true). To see this, we note as per lines $13$ and $14$ of \cref{alg:newcert} that the worst-case sum will occur if some patch location is vulnerable for every false negative. Because summation is done over the set of false negatives, this implies the worst-case sum is $FN$.
\end{proof}

We also provide a formal definition for the concept of a \emph{vulnerability status array}; recall that this array extends the certification procedure for a single-label CDPA (\cref{subsubsection:singlelabellocal}). We leverage similar notation as \cref{eq:linfty}. 

\begin{definition}[Vulnerability status array]
\label{def:vulstatusarray}
    Suppose we have datapoint $(\mathbf{x}, y)$, a single-label classifier $\mathbb{F}_s: \mathcal{X} \rightarrow \{1, 2, \dots, c\}$, a certification procedure $SL\text{-}CERT$ with security parameters $\sigma$ from a single-label CDPA, and patch locations $\mathcal{R}$. Then we define the \emph{vulnerability status array} $\bm{\lambda} := SL\text{-}CERT_{[\mathbb{F}_s, \sigma]}(\mathbf{x}, y, \mathcal{R}) \in \{0, 1\}^{|\mathcal{R}|}$ such that if $\bm{\lambda}[\mathbf{r}] = 1$ for a patch location $\mathbf{r} \in \mathcal{R}$ then\footnote{For the term $\bm{\lambda}[\mathbf{r}] = 1$ we slightly abuse notation and use $\mathbf{r}$ to refer to the index associated with the patch location}
    \begin{equation*}
      SL\text{-}INFER_{[\mathbb{F}_s, \sigma]}(\mathbf{r} \circ \mathbf{x} + (\mathbf{1} - \mathbf{r}) \circ \mathbf{x'}) = y \quad \forall \mathbf{x}' \in \mathcal{X}
    \end{equation*}
\end{definition}
\noindent
Essentially, the vulnerability status array $\bm{\lambda}$ denotes the certification status of individual patch locations $\mathbf{r} \in \mathcal{R}$.

We can now prove \cref{thm:histcorrect}. For convenience we re-state the theorem.

\histcorrect*

\begin{proof}
We will demonstrate the correctness and tightness of the new bound $FN_{new}$ proposed in \cref{alg:newcert}. We first note as per \cref{lemma:histtight} that $FN_{new} \leq FN$; this ensures that the new bound will be stronger than \cref{alg:kclasscert}. In the case with equality $FN_{new} = FN$, correctness is guaranteed by \cref{thm:patchdemuxcorrectnessclass}. We thus focus on the case with strict inequality $FN_{new} < FN$.

Define $\mathbf{r_{opt}} \in \mathcal{R}$ as the patch location which induces the maximum number of false negatives on line $14$ of \cref{alg:newcert}. By assumption, a total of $FN - FN_{new} > 0$ false negatives will have contributed a value of $0$ to the sum $fnTotal[\mathbf{r_{opt}}]$ on line $13$. Consider an arbitrary such class $i^* \in \{1, 2, \dots, c\}$. Because the $fnCertFails$ value for this class at patch location $\mathbf{r_{opt}}$ is $0$, on line $10$ we must have for $\bm{\lambda} := SL\text{-}CERT_{[\mathbb{F}[i^*], \sigma]}(\mathbf{x}, \mathbf{y}[i^*], \mathcal{R})$
\begin{equation*}
    \bm{\lambda}[\mathbf{r_{opt}}] = 1
\end{equation*}
As per \cref{def:vulstatusarray}, this means that we will have 
\begin{equation*}
    SL\text{-}INFER_{[\mathbb{F}[i^*], \sigma]}(\mathbf{r_{opt}} \circ \mathbf{x} + (\mathbf{1} - \mathbf{r_{opt}}) \circ \mathbf{x'}) = 1 \quad \forall \mathbf{x'} \in \mathcal{X}
\end{equation*}
In other words, $SL\text{-}INFER$ will be robust against any patch attack contained in location $\mathbf{r_{opt}} \in \mathcal{R}$. Because the patch must be placed at the optimal location $\mathbf{r_{opt}}$, this implies that \cref{alg:inference_patchdemux} will return the correct prediction for class $i^*$ as desired. Overall, each of the $FN - FN_{new}$ classes will now be certified true positives instead of false negatives, and thus the new bounds from \cref{alg:newcert} will be correct.
\end{proof}

\section{Double-masking Algorithm from PatchCleanser}
\label{section:appendix_patchcleanser}

In this section, we provide a brief outline of the double-masking algorithm from the PatchCleanser defense and how it integrates into the PatchDEMUX framework; recall from \cref{subsection:methodology} that PatchCleanser is the current SOTA single-label CDPA. For more details, we direct the reader to the original reference by \citet{xiang_patchcleanser_2022}.

\subsection{Double-masking overview}
\label{subsection:doublemaskingoverview_PC}
At a glance, the double-masking algorithm works by curating a specialized set of masks, $\mathcal{M} \subseteq \{0, 1\}^{w \times h}$, to recover the output label $y \in \{1, 2, \dots, c\}$ for certifiable input images $\mathbf{x} \in \mathcal{X}$ \cite{xiang_patchcleanser_2022}. More specifically, these masks satisfy the following \emph{$\mathcal{R}$-covering} property from \citet{xiang_patchcleanser_2022}.
\begin{definition}[$\mathcal{R}$-covering]
\label{def:rcovering}
  A mask set $\mathcal{M}$ is $\mathcal{R}$-covering if, for any patch in the patch region set $\mathcal{R}$, at least one mask from the mask set $\mathcal{M}$ can cover the entire patch, i.e., 
  \begin{equation*}
      \forall \mathbf{r} \in \mathcal{R}, \exists \mathbf{m} \in \mathcal{M} \quad s.t. \quad \mathbf{m}[i, j] \leq \mathbf{r}[i, j], \forall (i, j)
  \end{equation*}
\end{definition}
\noindent
Here $\mathcal{R}$ refers to the set of patch locations from \cref{eq:linfty}, and $\mathcal{M}$ represents binary matrices where elements inside the mask are $0$ and elements outside the mask are $1$ \cite{xiang_patchcleanser_2022}. Given an input image size $n_1 \times n_2$, an upper estimate on patch size\footnote{PatchCleanser provides an option to specify the patch size for each axis; we simplify the notation here for convenience} $p$, and number of desired masks $k_1 \times k_2$, a procedure from \citet{xiang_patchcleanser_2022} can readily create a mask set $\mathcal{M}$ with stride length $s_1 \times s_2$ and mask size $m_1 \times m_2$ which is $\mathcal{R}$-covering. The patch size $p$ and mask number $k_1 \times k_2$ serve as security parameters, where the former corresponds to the threat level of $\mathcal{R}$ (i.e., larger patches will necessitate larger masks) and the latter represents a computational budget (i.e., more masks will require more checks to be performed) \cite{xiang_patchcleanser_2022}.

Once the $\mathcal{R}$-covering mask set $\mathcal{M}$ is generated, the double-masking inference procedure removes the patch by selectively occluding the image $\mathbf{x} \in \mathcal{X}$ with mask pairs $\mathbf{m_0}, \mathbf{m_1} \in \mathcal{M} \times \mathcal{M}$. Correctness is verified through the associated certification procedure, which checks if predictions on $\mathbf{x}$ are preserved across all possible mask pairs \cite{xiang_patchcleanser_2022}. 

\subsection{Double-masking inference procedure}
\label{subsection:doublemaskinginference_PC}

\begin{algorithm}[!h]
\caption{\textit{The double-masking inference procedure from PatchCleanser \cite{xiang_patchcleanser_2022}}}\label{alg:inference_doublemasking_PC}
\begin{algorithmic}[1]

\item[] \textbf{Input:} Image $\mathbf{x} \in \mathcal{X}$, single-label classifier $\mathbb{F}_s: \mathcal{X} \rightarrow \{1, 2, \dots, c\}$, $\mathcal{R}$-covering mask set $\mathcal{M}$
\item[] \textbf{Output:} Prediction $\hat{y} \in \{1, 2, \dots, c\}$
\Procedure{DoubleMaskingInfer}{$\mathbf{x}, \mathbb{F}_s, \mathcal{M}$}
\State $\hat{y}_{maj}, \mathcal{P}_{dis} \gets $\Call{MaskPred}{$\mathbf{x}, \mathbb{F}_s, \mathcal{M}$} \Comment{First-round}
\If {$\mathcal{P}_{dis} = \emptyset$}
    \State \Return $\hat{y}_{maj}$ \Comment{Case I: agreed prediction}
\EndIf

\For{each $(\mathbf{m_{dis}}, \hat{y}_{dis}) \in \mathcal{P}_{dis}$} \Comment{Second round}
    \State $\hat{y}', \mathcal{P}' \gets$ \Call{MaskPred}{$\mathbf{x} \circ \mathbf{m_{dis}}, \mathbb{F}_s, \mathcal{M}$}
    \If {$\mathcal{P}' = \emptyset$}
        \State \Return $\hat{y}_{dis}$ \Comment{Case II: disagreer pred.}
    \EndIf
\EndFor
\State \Return $\hat{y}_{maj}$ \Comment{Case III: majority prediction}
\EndProcedure
\[\]
\item[] \textbf{Input:} Image $\mathbf{x} \in \mathcal{X}$, single-label classifier $\mathbb{F}_s: \mathcal{X} \rightarrow \{1, 2, \dots, c\}$, $\mathcal{R}$-covering mask set $\mathcal{M}$
\item[] \textbf{Output:} Majority prediction $\hat{y}_{maj} \in \{1, 2, \dots, c\}$, disagreer masks $\mathcal{P}_{dis}$
\Procedure{MaskPred}{$\mathbf{x}, \mathbb{F}_s, \mathcal{M}$}
    \State $\mathcal{P} \gets \emptyset$ \Comment{A set for mask-prediction pairs}
    \For{$\mathbf{m} \in \mathcal{M}$} \Comment{Enumerate every mask $\mathbf{m}$}
        \State $\hat{y} \gets \mathbb{F}_s(\mathbf{x} \circ \mathbf{m})$ \Comment{Evaluate masked prediction}
        \State $\mathcal{P} \gets \mathcal{P} \bigcup \{(\mathbf{m}, \hat{y})\}$ \Comment{Update set $\mathcal{P}$}
    \EndFor
    \State $\hat{y}_{maj} \gets \text{argmax}_{y^{*}}|\{(\mathbf{m}, \hat{y}) \in \mathcal{P} | \hat{y} = y^{*}\}|$ \Comment{Majority}
    \State $\mathcal{P}_{dis} \gets \{(\mathbf{m}, \hat{y}) \in \mathcal{P} | \hat{y} \neq \hat{y}_{maj}\}$ \Comment{Disagreers}
    \State \Return{$\hat{y}_{maj}, \mathcal{P}_{dis}$}
\EndProcedure

\end{algorithmic}
\end{algorithm}

The double-masking inference procedure from \citet{xiang_patchcleanser_2022} is outlined in \cref{alg:inference_doublemasking_PC}. It works by running up to two rounds of masking on the input image $\mathbf{x} \in \mathcal{X}$.  In each round, the single-label classifier $\mathbb{F}_s: \mathcal{X} \rightarrow \{1, 2, \dots, c\}$ is queried on copies of $\mathbf{x}$ which have been augmented by masks $\mathbf{m} \in \mathcal{M}$ \cite{xiang_patchcleanser_2022}.

\begin{itemize}
    \item \emph{First-round masking:} The classifier runs $\mathbb{F}_s(\mathbf{m} \circ \mathbf{x})$ for every mask $\mathbf{m} \in \mathcal{M}$ (line $2$). If there is consensus, this is returned as the overall prediction (line $4$); the intuition is that a clean image with no patch will be predicted correctly regardless of the mask present \cite{xiang_patchcleanser_2022}. Otherwise, the minority/``disagreer'' predictions trigger a second-round of masking (line $6$). This is done to determine whether to trust the majority prediction $\hat{y}_{maj}$ or one of the disagreers \cite{xiang_patchcleanser_2022}.
    
    \item \emph{Second-round masking:} For each disagreer mask $\mathbf{m_{dis}}$, the classifier runs $\mathbb{F}_s(\mathbf{x} \circ \mathbf{m_{dis} \circ \mathbf{m}} )$ for every mask $\mathbf{m} \in \mathcal{M}$ to form \emph{double-mask predictions} \cite{xiang_patchcleanser_2022}. If there is consensus, the disagreer label $\hat{y}_{dis}$ associated with $\mathbf{m_{dis}}$ is returned as the overall prediction (lines $6-10$). The intuition is that consensus is likely to occur if $\mathbf{m_{dis}}$ successfully covered the patch \cite{xiang_patchcleanser_2022}. Otherwise, $\mathbf{m_{dis}}$ is ignored and the next available disagreer mask is considered; the assumption here is that $\mathbf{m_{dis}}$ failed to cover the patch \cite{xiang_patchcleanser_2022}. Finally, if none of the disagreer masks feature consensus the majority label $\hat{y}_{maj}$ from the first-round is returned instead (line $12$).   
\end{itemize}
\noindent
A key property of this method is that it is architecture agnostic and can be integrated with any single-label classifier \cite{xiang_patchcleanser_2022}.

\subsection{Double-masking certification procedure}
\label{subsection:doublemaskingcertification_PC}

\begin{algorithm}[!h]
\caption{\textit{The double-masking certification procedure from PatchCleanser \cite{xiang_patchcleanser_2022}}}\label{alg:cert_doublemasking_PC}
\begin{algorithmic}[1]
\item[] \textbf{Input:} Image $\mathbf{x} \in \mathcal{X}$, ground-truth $y \in \{1, 2, \dots, c\}$, single-label classifier $\mathbb{F}_s: \mathcal{X} \rightarrow \{1, 2, \dots, c\}$, patch locations $\mathcal{R}$, $\mathcal{R}$-covering mask set $\mathcal{M}$
\item[] \textbf{Output:} Overall certification status of $(\mathbf{x}, y)$, vulnerability status array $\bm{\lambda} \in \{0, 1\}^{|\mathcal{M}|}$ 
\Procedure{DoubleMaskingCert}{$\mathbf{x}, y, \mathbb{F}_s, \mathcal{R}, \mathcal{M}$}
    \State $certVal \gets 1$
    \State $\bm{\lambda} \gets [1]^{|\mathcal{M}|}$
    \If{$\mathcal{M}$ is not $\mathcal{R}$-covering} \Comment{Insecure mask set}
        \State \Return{$0$, $[0]^{|\mathcal{M}|}$}
    \EndIf
    \For{every $(\mathbf{m_0}, \mathbf{m_1}) \in \mathcal{M} \times \mathcal{M}$}
        \State $\hat{y}' \gets \mathbb{F}_s(\mathbf{x} \circ \mathbf{m_0} \circ \mathbf{m_1})$ \Comment{Two-mask prediction}
        \If{$\hat{y'} \neq y$}
            \State $certVal \gets 0$ \Comment{Input possibly vulnerable}
            \State $\bm{\lambda}[\mathbf{m_0}], \bm{\lambda}[\mathbf{m_1}] \gets 0, 0$ \Comment{Vulnerable masks}
        \EndIf
    \EndFor
    \State \Return{$certVal, \bm{\lambda}$}
\EndProcedure

\end{algorithmic}
\end{algorithm}

The double-masking certification procedure from \citet{xiang_patchcleanser_2022} is outlined in \cref{alg:cert_doublemasking_PC}; we extend the original version to additionally return a vulnerability status array $\bm{\lambda}$. It works by first ensuring that the mask set $\mathcal{M}$ is $\mathcal{R}$-covering (line $4$); otherwise, no guarantees on robustness can be made. Then, during the $for$ loop on lines $7-13$ the procedure computes $\mathbb{F}_s(\mathbf{x} \circ \mathbf{m_0} \circ \mathbf{m_1})$ for every possible mask pair $\mathbf{m_0}, \mathbf{m_1} \in \mathcal{M} \times \mathcal{M}$ \cite{xiang_patchcleanser_2022}. If all of the predictions are the label $y$, then $(\mathbf{x}, y)$ is certifiable and $certVal$ is set to $1$; recall from \cref{def:certdefensesslpatches} that this implies that the inference procedure \cref{alg:inference_doublemasking_PC} will be correct regardless of an attempted patch attack. Otherwise, $certVal$ is set to $0$ and the $\bm{\lambda}$ array is updated to reflect vulnerable points.

The correctness of $certVal$ is guaranteed by the following theorem. Essentially, if predictions across all possible mask pairs are correct, it ensures that each of the three cases in \cref{alg:inference_doublemasking_PC} will work as intended \cite{xiang_patchcleanser_2022}.
\begin{theorem}
\label{thm:doublemaskingcorrectness_PC}
Suppose we have an image data point $(\mathbf{x}, y)$, a single-label classification model $\mathbb{F}_s: \mathcal{X} \rightarrow \{1, 2, \dots, c\}$, a patch threat model $S_{\mathbf{x}, \mathcal{R}}$, and a $\mathcal{R}$-covering mask set $\mathcal{M}$. If $\mathbb{F}_s(\mathbf{x} \circ \mathbf{m_0} \circ \mathbf{m_1}) = y$ for all $\mathbf{m_0}, \mathbf{m_1} \in \mathcal{M} \times \mathcal{M}$, then \cref{alg:inference_doublemasking_PC} will always return a correct label.
\end{theorem}
\begin{proof}
    This theorem is proved in \citet{xiang_patchcleanser_2022}.
\end{proof}
We next consider the vulnerability status array $\bm{\lambda} \in \{0, 1\}^{|\mathcal{M}|}$ returned by \cref{alg:cert_doublemasking_PC}. Notice that the length of the array is $|\mathcal{M}|$ rather than $|\mathcal{R}|$; this is a helpful consequence of the $\mathcal{R}$-covering property of the mask set $\mathcal{M}$, which ensures that every patch location $\mathbf{r} \in \mathcal{R}$ will be contained in at least one of the masks $\mathbf{m} \in \mathcal{M}$. As such, an implementation-level abstraction is possible for PatchCleanser where each element $\bm{\lambda}[\mathbf{m}]$ summarizes the vulnerability status for all patch locations contained within the mask $\mathbf{m} \in \mathcal{M}$. The correctness of this construction can be demonstrated through the following lemma. 
\begin{lemma}
\label{lemma:doublemasking_lambda}
Suppose we have an image data point $(\mathbf{x}, y)$, a single-label classification model $\mathbb{F}_s: \mathcal{X} \rightarrow \{1, 2, \dots, c\}$, a patch threat model $S_{\mathbf{x}, \mathcal{R}}$, and a $\mathcal{R}$-covering mask set $\mathcal{M}$. Then the array $\bm{\lambda} \in \{0, 1 \}^{|\mathcal{M}|}$ returned by \cref{alg:cert_doublemasking_PC} will be a valid vulnerability status array that satisfies \cref{def:vulstatusarray}.
\end{lemma}
\begin{proof}
    Define $\mathcal{R^*} \subseteq \mathcal{R}$ as the set of patch locations contained in an arbitrary mask $\mathbf{m^*} \in \mathcal{M}$. To demonstrate the validity of $\bm{\lambda}$, we need to show that $\bm{\lambda}[\mathbf{m^*}] = 1$ implies \cref{alg:inference_doublemasking_PC} will be protected from all attacks located in $\mathcal{R^*}$. To do so, we first note that we will only have $\bm{\lambda}[\mathbf{m^*}] = 1$ in \cref{alg:cert_doublemasking_PC} if $\mathbb{F}_s(\mathbf{x} \circ \mathbf{m^*} \circ \mathbf{m}) = y$ for all $\mathbf{m} \in \mathcal{M}$; otherwise, $\bm{\lambda}[\mathbf{m^*}]$ would have been marked with $0$ at some point. 

    We can use this robustness property to guarantee correctness in \cref{alg:inference_doublemasking_PC}. Suppose we have an arbitrary patch attack with a location in $\mathcal{R^*}$ and that $\bm{\lambda}[\mathbf{m^*}] = 1$. In the first-round masking stage the attack will be completely covered by the mask $\mathbf{m^*}$ (due to the $\mathcal{R}$-covering property) and form the masked image $\mathbf{x} \circ \mathbf{m^*} \in \mathcal{X}$. Note that this is the same as the image $\mathbf{x} \circ \mathbf{m^*} \circ \mathbf{m^*} \in \mathcal{X}$; therefore, the robustness property from above will guarantee that $\mathbb{F}_s(\mathbf{x} \circ \mathbf{m^*}) = y$. We have thus shown that the correct prediction will be represented at least once in the first-round, leaving three possible scenarios.
    \begin{itemize}
        \item \emph{Scenario \#1 (consensus):} In this scenario, the classifier returns the correct prediction $y$ for every first-round mask. Then, line $4$ of \cref{alg:inference_doublemasking_PC} will ensure that $y$ is correctly returned as the overall prediction.
        \item \emph{Scenario \#2 (majority of masks are correct):} In this scenario, the classifier returns the correct prediction $y$ for the majority of first-round masks. The set of disagreer masks, $\mathcal{M}_{dis} \subseteq \mathcal{M}$, will thus run a second round of masking. This eventually requires computing $\mathbb{F}_s(\mathbf{x} \circ \mathbf{m_{dis}} \circ \mathbf{m^*})$ for each $\mathbf{m_{dis}} \in \mathcal{M}_{dis}$. By leveraging symmetry and the robustness property from earlier, these are all guaranteed to return the correct prediction $y$. Therefore, none of the disagreer masks will have consensus in the second-round, and line $12$ of \cref{alg:inference_doublemasking_PC} will ensure that $y$ is correctly returned as the overall prediction.
        \item \emph{Scenario \#3 (minority of masks are correct):} In this scenario, the classifier returns the correct prediction $y$ for a minority of first-round masks. This implies that $\mathbf{m^*}$ will be a disagreer mask. During the second round of masking, the robustness property from earlier will ensure that $\mathbb{F}_s(\mathbf{x} \circ \mathbf{m^*} \circ \mathbf{m}) = y$ for each $\mathbf{m} \in \mathcal{M}$. Therefore, we will have consensus in the second round of $\mathbf{m^*}$. Because disagreers with incorrect predictions will fail to have consensus (i.e., using the logic from \emph{Scenario \#2}), line $9$ of \cref{alg:inference_doublemasking_PC} will ensure that $y$ is correctly returned as the overall prediction.
    \end{itemize}
    
    Overall, we conclude that \cref{alg:inference_doublemasking_PC} will return the correct prediction $y$. We have thus shown that $\bm{\lambda}[\mathbf{m^*}] = 1$ implies \cref{alg:inference_doublemasking_PC} will be protected from any arbitrary patch attack located in $\mathcal{R^*}$, as desired.
\end{proof}

\subsection{Integration with PatchDEMUX}
\label{subsection:patchdemux_pc}
To integrate PatchCleanser into the PatchDEMUX framework, we first generate a $\mathcal{R}$-covering set of masks $\mathcal{M}$; the mask set $\mathcal{M}$ essentially serves as a holistic representation of the security parameters $\sigma$. We then incorporate \cref{alg:inference_doublemasking_PC} into the PatchDEMUX inference procedure (\cref{alg:inference_patchdemux}) and \cref{alg:cert_doublemasking_PC} into the PatchDEMUX certification procedure (\cref{alg:kclasscert}). Finally, we use the location-aware certification method (\cref{alg:newcert}) with the vulnerability status arrays expressed in terms of masks. 

\section{Further Details on Evaluation Metrics}
\label{section:appendix_evalmetrics}
In this section, we discuss the evaluation metrics from \cref{section:mainresults} and \cref{section:securityparam} in more detail.

\subsection{Threshold analysis}
\label{subsection:threshold}
We evaluate multi-label classifiers by computing precision and recall metrics over a variety of different thresholds; classes with output higher than the threshold are predicted $1$, otherwise $0$. This helps establish a large set of evaluation data from which to build precision-recall plots. We start by evaluating a set of \emph{standard thresholds}:
\begin{equation*}
    T_{standard} := \{0.0, 0.1, 0.2, 0.3, 0.4, 0.5, 0.6, 0.7, 0.8, 0.9\}
\end{equation*}
We then evaluate a set of \emph{high-value thresholds}. This helps fill out the \emph{low recall-high precision} region of a precision-recall curve:
\begin{equation*}
    T_{high} := \{0.91, 0.92, 0.93, 0.94, 0.95, 0.96, 0.97, 0.98, 0.99\}
\end{equation*}
We next evaluate a set of \emph{very high-value thresholds}. These evaluations provide points at which recall is close to $0\%$:
\begin{equation*}
    T_{very \: high} := \{0.999, 0.9999, 0.99999\}
\end{equation*}
Finally, we evaluate a set of mid-value thresholds. These help to smoothen out a precision-recall curve:
\begin{equation*}
    T_{mid} := T_{mid \: 1} \cup T_{mid \: 2} \cup T_{mid \: 3} \cup T_{mid \: 4}
\end{equation*}
where 
\begin{equation*}
\begin{split}
    T_{mid \: 1} := \{0.5 + 0.02 \cdot t : t \in \{1, 2, 3, 4\}\} \\
    T_{mid \: 2} := \{0.6 + 0.02 \cdot t : t \in \{1, 2, 3, 4\}\} \\
    T_{mid \: 3} := \{0.7 + 0.02 \cdot t : t \in \{1, 2, 3, 4\}\} \\
    T_{mid \: 4} := \{0.8 + 0.02 \cdot t : t \in \{1, 2, 3, 4\}\}
\end{split}
\end{equation*}

For ViT-based models specifically, we found that the \emph{low precision-high recall} region of a precision-recall curve does not readily appear if we limit evaluation to the thresholds outlined above. We thus further evaluate the following set of \emph{low-value} thresholds for ViT-based models:
\begin{equation*}
    T_{low} := \{5\cdot 10^{-5}, 10^{-4}, 5\cdot 10^{-4}, 10^{-3}, 5\cdot 10^{-3}, 0.01, 0.05\}
\end{equation*}

In order to obtain precision values at key recall levels (i.e., $25\%$, $50\%$, $75\%$), we can perform linear interpolation between relevant recall bounds. However, recall values computed using the thresholds above are often not close enough to these target values. To this end, we use an iterative bisection scheme to find overestimated and underestimated bounds within $0.5$ points of the target recalls. The precision values are then calculated by linearly interpolating between these bounds.

\subsection{Computing average precision}
\label{subsection:aucmetrics}

In order to compute an approximation for average precision, we leverage the area-under-the-curve (AUC) of the associated precision-recall curves. However, in practice the threshold analysis from \cref{subsection:threshold} can result in different leftmost points for the precision-recall curves. In order to enforce consistency, we fix the leftmost points for each precision-recall plot at exactly $25\%$ recall. Then, the AUC is computed using the trapezoid sum technique and normalized by a factor of $0.75$ (i.e., the ideal precision-recall curve). Note that we pick $25\%$ recall because a few evaluations under this value demonstrate floating-point precision errors (i.e., the required threshold is too high).

\section{Resnet Architecture Analysis}
\label{section:appendix_resnet}
\begin{table*}[!ht]
    \centering
    \caption{\textit{PatchDEMUX performance with Resnet architecture on the MS-COCO 2014 validation dataset. Precision values are evaluated at key recall levels along with the approximated average precision. We assume the patch attack is at most $2\%$ of the image area and use a computational budget of $6 \times 6$ masks.}}
    \subcaptionbox{\textit{Clean setting precision values}\label{tab:resnetbesttable_clean}}{
    \begin{tabular}{lcccc}
        \toprule
        \textbf{Architecture} & \multicolumn{4}{c}{Resnet} \\
        \cmidrule(l){1-1}\cmidrule(l){2-5}
        \textbf{Clean recall} & $25\%$ & $50\%$ & $75\%$ & $AP$ \\
        \midrule
        \textit{Undefended} & 99.832 & 99.425 & 92.341 & 87.608\\
        \textit{Defended} & 99.835 & 98.257 & 80.612 & 81.031 \\
        \bottomrule
    \end{tabular}
    }
    \subcaptionbox{\textit{Certified robust setting precision values}\label{tab:resnetbesttable_robust}}{
    \begin{tabular}{lcccc}
        \toprule
        \textbf{Architecture} & \multicolumn{4}{c}{Resnet} \\
        \cmidrule(l){1-1}\cmidrule(l){2-5}
        \textbf{Certified recall} & $25\%$ & $50\%$ & $75\%$ & $AP$ \\
        \midrule
        \textit{Certified robust} & 86.696 & 40.190 & 20.959 & 34.859 \\
        \textit{Location-aware} & 87.950 & 44.373 & 23.202 & 37.544\\
        \bottomrule
    \end{tabular}
    }
    \label{tab:resnet_bestdata}
\end{table*}

\begin{figure*}[!ht]
    \centering
    \subcaptionbox{\textit{Clean setting precision-recall curves}\label{fig:resnetbestprecrecall_clean}}{\includegraphics[width=0.4\textwidth]{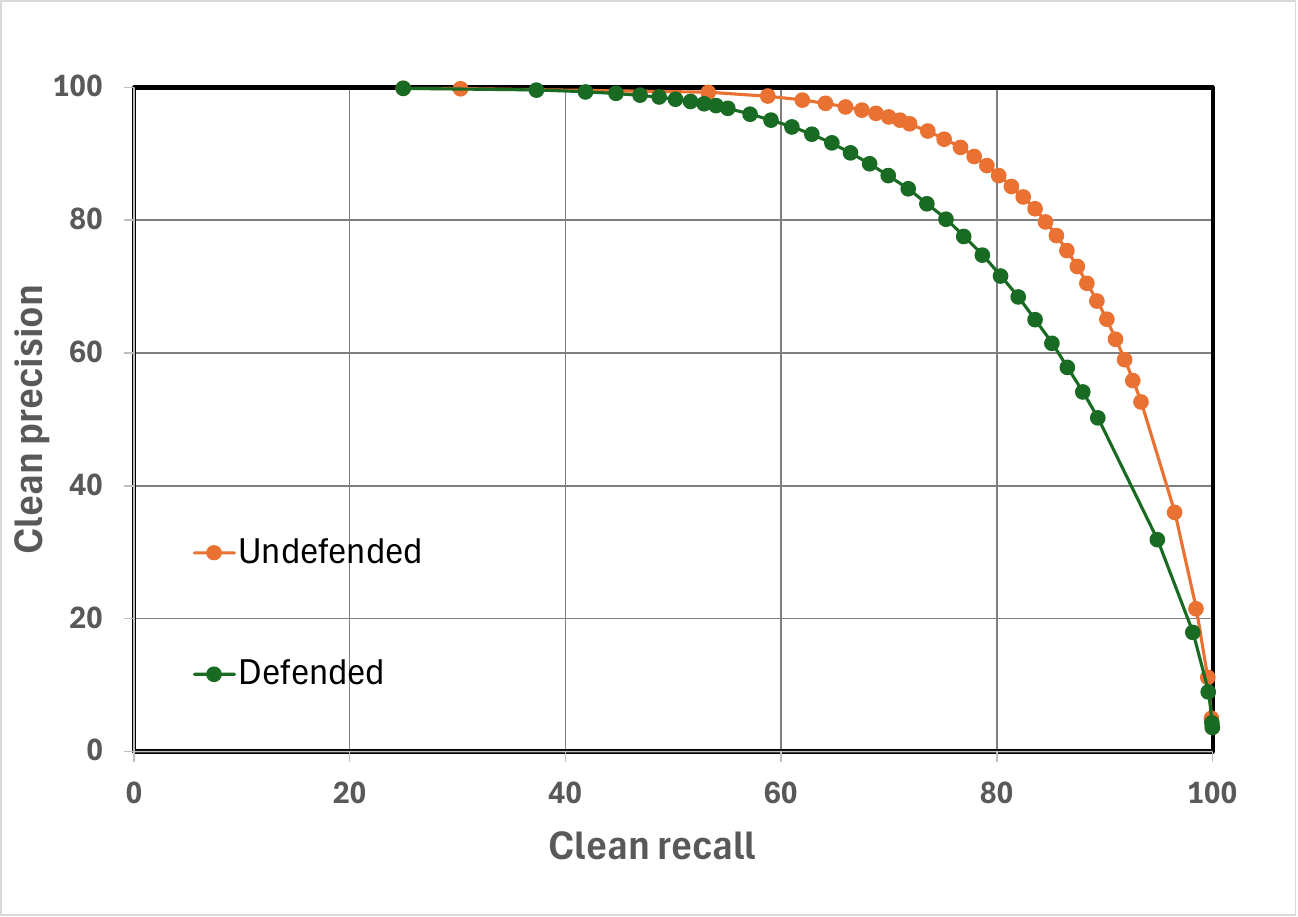}}
    \subcaptionbox{\textit{Certified robust setting precision-recall curves}\label{fig:resnetbestprecrecall_robust}}{\includegraphics[width=0.4\textwidth]{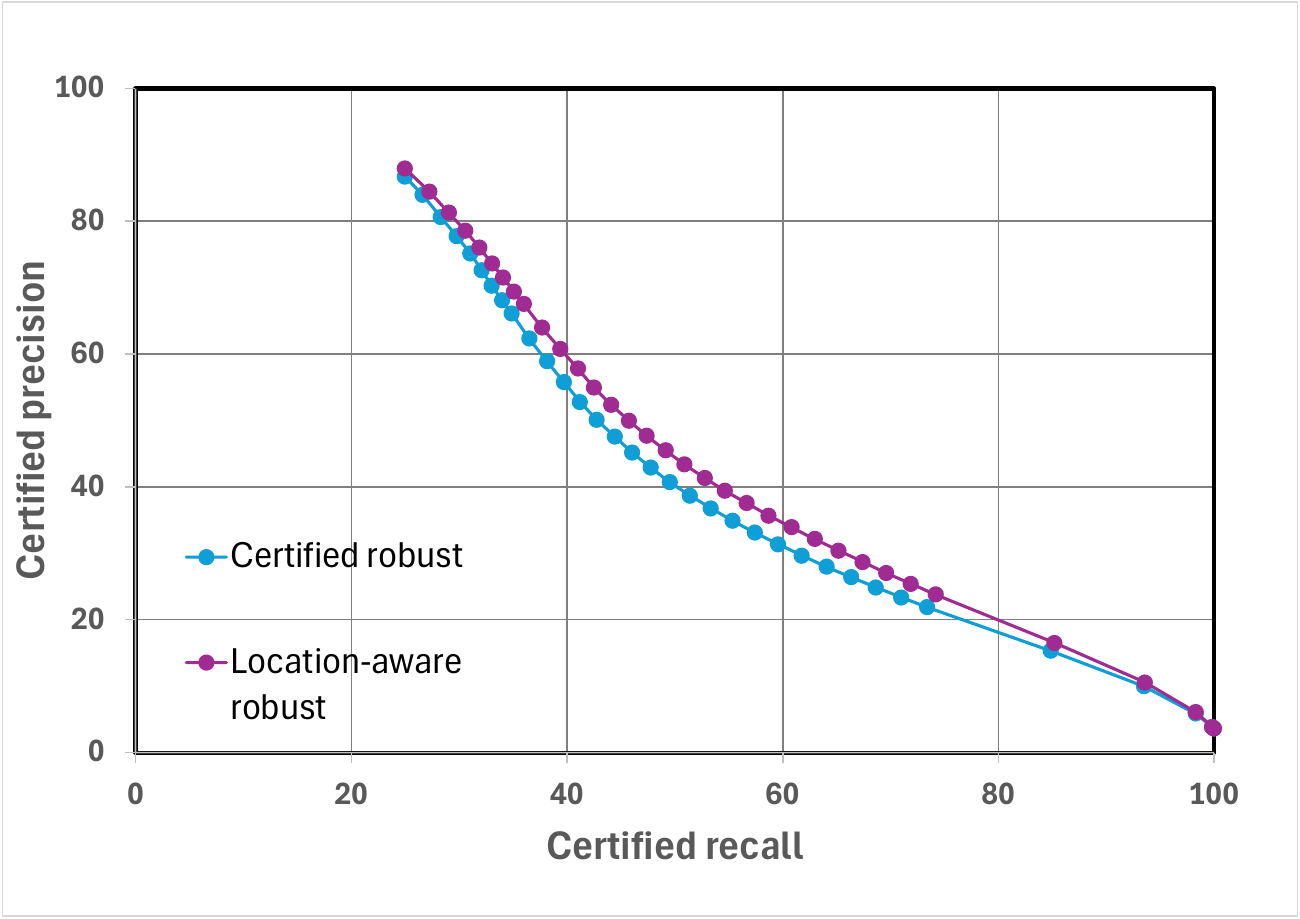}}
    \caption{\textit{PatchDEMUX precision-recall curves with Resnet architecture over the MS-COCO 2014 validation dataset. We consider the clean and certified robust evaluation settings. We assume the patch attack is at most $2\%$ of the image area and use a computational budget of $6 \times 6$ masks.}}
    \label{fig:resnetbestprecrecall}
\end{figure*}

In this section, we report results for PatchDEMUX while using the Resnet architecture \cite{ben-baruch_asymmetric_2021}; we leverage the same defense fine-tuning routine as \cref{subsection:bestresults} to achieve stronger performance. Experiments are done on the MS-COCO 2014 validation dataset. The precision values associated with key recall levels are in \cref{tab:resnet_bestdata}. \cref{fig:resnetbestprecrecall} features precision-recall plots, while AP values are present in \cref{tab:resnet_bestdata}. 

We find that the Resnet and ViT architectures show similar qualitative trends. For instance, defended clean performance is close to undefended performance across a variety of thresholds (see \cref{fig:resnetbestprecrecall_clean} and \cref{tab:resnetbesttable_clean}). The Resnet-based variant of PatchDEMUX also achieves non-trivial robustness, with a certified average precision of $37.544\%$. In general, the precision-recall curves for the two architectures are similar in shape across all four evaluation settings. 

Despite these similarities, the ViT model consistently outperforms the Resnet model. More specifically, the ViT-based variant of PatchDEMUX provides a ${\sim}4$ point boost to clean AP in the defended clean setting and a ${\sim}7$ point boost to certified AP in the two certified robust settings (see \cref{tab:bestdata}). This improvement might be attributable to the training procedure of vision transformers, which involves a masking process that is similar in concept to PatchCleanser's double-masking algorithm \cite{dosovitskiy_image_2021, xiang_patchcleanser_2022}.

\section{Location-aware Certification Analysis}
\label{section:appendix_locablation}

\begin{table}[!ht]
    \centering
    \caption{\textit{ViT-based PatchDEMUX performance with different location-aware attackers. Experiments performed on the MS-COCO 2014 validation dataset. Precision values are evaluated at key recall levels along with the approximated average precision. We assume the patch attack is at most $2\%$ of the image area and use a computational budget of $6 \times 6$ masks.}}
    \resizebox{\linewidth}{!}{\begin{tabular}{lcccc}
        \toprule
        \textbf{Architecture} & \multicolumn{4}{c}{ViT} \\
        \cmidrule(l){1-1}\cmidrule(l){2-5}
        \textbf{Certified recall} & $25\%$ & $50\%$ & $75\%$ & $AP$ \\
        \midrule
        \textit{FP attacker} & 95.724 & 62.132 & 33.112 & 49.474 \\
        \textit{FN attacker} & 95.971 & 58.158 & 27.199 & 45.951 \\
        \textit{Location-aware robust} & 95.670 & 56.038 & 26.375 & 44.902 \\
        \textit{Certified robust} & 95.369 & 50.950 & 22.662 & 41.763\\
        \bottomrule
    \end{tabular}}
    \label{tab:histdata}
\end{table}

\begin{figure}[!ht]
    \centering
    \includegraphics[width=0.4\textwidth]{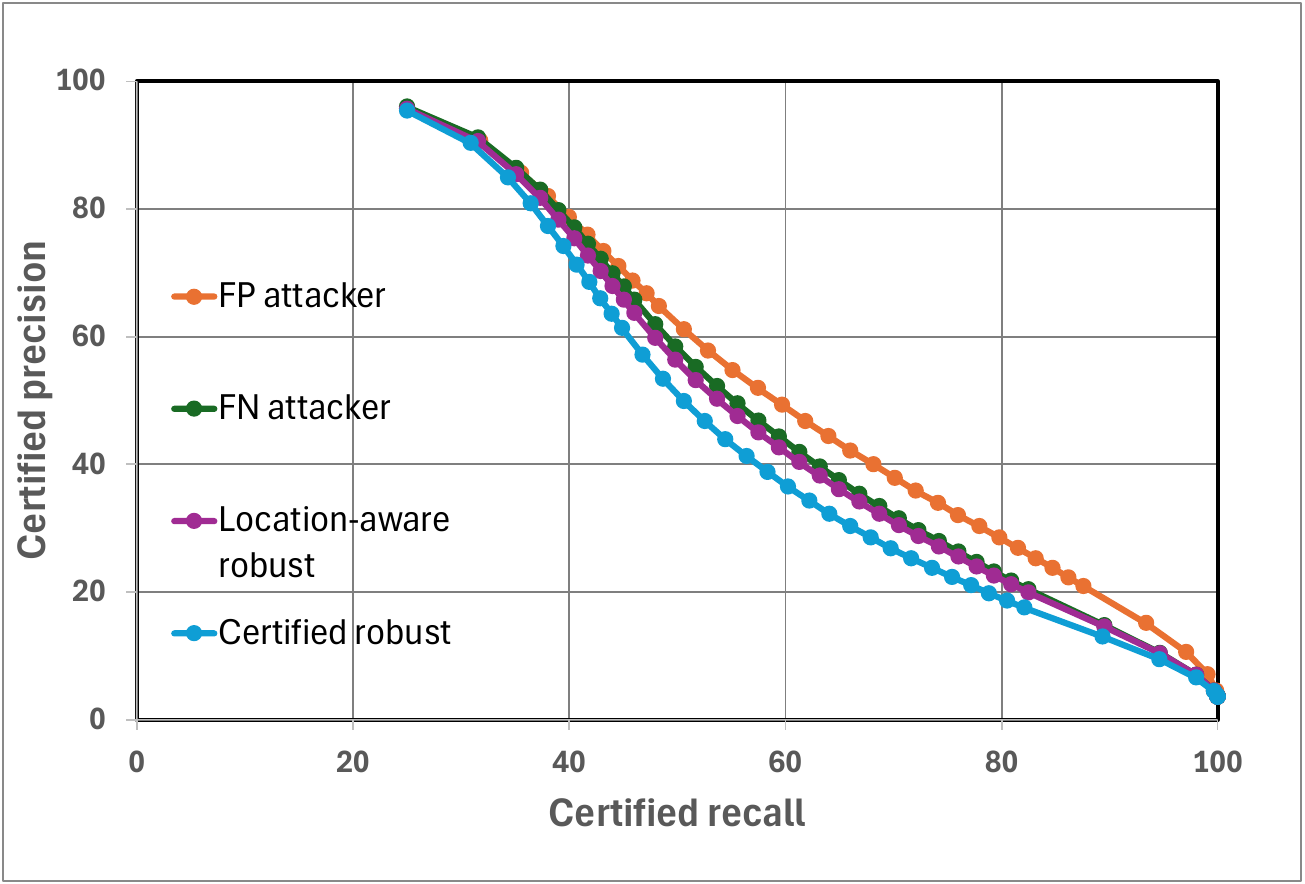}
    \caption{\textit{ViT-based PatchDEMUX precision-recall curves with different location-aware attackers. Experiments performed on the MS-COCO 2014 validation dataset. The baseline certified robust evaluation setting is included for comparison. We assume the patch attack is at most $2\%$ of the image area and use a computational budget of $6 \times 6$ masks.}}
    \label{fig:histprecrecall}
\end{figure}

In this section we investigate the location-aware certification approach from \cref{subsection:newcertmethod} in more detail.

\subsection{Attack vectors}
\label{subsection:histdiffattacks}
Based on \cref{subsubsection:newcertproposal}, there are a couple different ways to evaluate the robustness provided by the location-aware method. 
\begin{itemize}
    \item \emph{$FN$ attacker:} Here, we only track vulnerability status arrays $\bm{\lambda}$ for false negatives. Intuitively, this corresponds to the optimal attacker from \cref{subsubsection:newcertproposal} constructing a patch with the sole intent of increasing false negatives (i.e., a $FN$ attack). In \cref{alg:newcert}, tie-breakers are decided by picking the location which induces more false positives.
    \item \emph{$FP$ attacker:} In this scenario we only track vulnerability status arrays $\bm{\lambda}$ for false positives. This corresponds to the optimal attacker from \cref{subsubsection:newcertproposal} constructing a patch with the sole intent of increasing false positives (i.e., a $FP$ attack). In the $FP$ version of \cref{alg:newcert}, tie-breakers are decided by picking the location which induces more false negatives.
\end{itemize}
We also consider ``worst case'' performance where we simultaneously determine the worst patch location for both false negatives and false positives. Note that these two locations do not have to be identical, and as a result this ``worst case'' performance is not necessarily realizable. However, we evaluate this approach because it represents the theoretical lower bound on robustness for \cref{alg:newcert} given an arbitrarily motivated attacker.

\subsection{Experiment results}
\label{subsection:histresults}
We now empirically compare the different attack vectors possible under location-aware certification. We consider the ViT architecture alone as it provides better performance compared to Resnet. In addition, we leverage the same pre-trained model checkpoints used in \cref{subsection:bestresults} for consistency. Experiments are done on the MS-COCO 2014 validation dataset. Precision values corresponding to different attackers are present in \cref{tab:histdata}, while precision-recall plots are in \cref{fig:histprecrecall}.

\textbf{Provable robustness improvements.} Regardless of the attack strategy employed, location-aware certification provides improved robustness compared to the baseline certified robust setting; this is expected due to \cref{thm:histcorrect}. Improvement is most notable in both the \emph{mid recall-mid precision} and \emph{high recall-low precision} sections of the precision-recall curve. Overall, the most favorable evaluation approach provides an ${\sim}8$ point increase in certified AP compared to the baseline. Despite these improvements, location-aware certification does not fundamentally change the shape of the robust precision-recall curve under any of the three attack settings.

\textbf{Asymmetric attack performance.} Interestingly, location-aware certification provides the strongest robustness guarantees under the $FP$ attack strategy. This is likely due to the asymmetric dependence of precision and recall metrics on false positives. Specifically, both metrics depend on false negatives\footnote{Precision indirectly depends on false negatives via the true positive count}, but the recall metric does not depend on false positives. This makes $FP$ attacks ``weaker'' relative to other methods.

\section{Defense Fine-tuning for PatchDEMUX}
\label{section:appendix_pretrainablation}

\begin{table*}[!ht]
    \centering
    \caption{\textit{PatchDEMUX performance with ViT architecture on the MS-COCO 2014 validation dataset when using different defense fine-tuning techniques. Precision values are evaluated at key recall levels along with the approximated average precision. We assume the patch attack is at most $2\%$ of the image area and use a computational budget of $6 \times 6$ masks.}}
    \subcaptionbox{\textit{Clean setting precision values}\label{tab:clean_pretrainingdata}}{
        \tabcolsep=0.15cm
        \resizebox{\linewidth}{!}{
        \begin{tabular}{lcccccccccccccccc}
            \toprule
            \textbf{Architecture} & \multicolumn{4}{c}{ViT (vanilla/no fine-tuning)} & \multicolumn{4}{c}{ViT (Random Cutout)} & \multicolumn{4}{c}{ViT (Greedy Cutout $6 \times 6$)} & \multicolumn{4}{c}{ViT (Greedy Cutout $3 \times 3$)}\\
            \cmidrule(l){1-1}\cmidrule(l){2-5}\cmidrule(l){6-9}\cmidrule(l){10-13}\cmidrule(l){14-17}
            \textbf{Recall} & $25\%$ & $50\%$ & $75\%$ & $AP$ & $25\%$ & $50\%$ & $75\%$ & $AP$ & $25\%$ & $50\%$ & $75\%$ & $AP$ & $25\%$ & $50\%$ & $75\%$ & $AP$ \\
            \midrule
            \textit{Undefended} & 99.940 & 99.749 & 96.265 & 91.449 & 99.770 & 99.642 & 95.951 & 90.900 & 99.930 & 99.704 & 96.141 & 91.146 & 99.930 & 99.736 & 95.973 & 90.903 \\
            \textit{Defended} & 99.930 & 99.138 & 85.757 & 83.776 & 99.858 & 99.224 & 87.273 & 85.028 & 99.894 & 99.223 & 87.764 & 85.276 & 99.900 & 99.230 & 87.741 & 85.271 \\
            \bottomrule
        \end{tabular}
        }
    }
    \par\bigskip
    \subcaptionbox{\textit{Certified robust setting precision values}\label{tab:robust_pretrainingdata}}{
        \tabcolsep=0.15cm
        \resizebox{\linewidth}{!}{
        \begin{tabular}{lcccccccccccccccc}
            \toprule
            \textbf{Architecture} & \multicolumn{4}{c}{ViT (vanilla/no fine-tuning)} & \multicolumn{4}{c}{ViT (Random Cutout)} & \multicolumn{4}{c}{ViT (Greedy Cutout $6 \times 6$)} & \multicolumn{4}{c}{ViT (Greedy Cutout $3 \times 3$)}\\
            \cmidrule(l){1-1}\cmidrule(l){2-5}\cmidrule(l){6-9}\cmidrule(l){10-13}\cmidrule(l){14-17}
            \textbf{Certified recall} & $25\%$ & $50\%$ & $75\%$ & $AP$ & $25\%$ & $50\%$ & $75\%$ & $AP$ & $25\%$ & $50\%$ & $75\%$ & $AP$ & $25\%$ & $50\%$ & $75\%$ & $AP$ \\
            \midrule
            \textit{Certified robust} & 90.767 & 38.490 & 20.846 & 35.003 & 94.192 & 47.548 & 24.603 & 40.975 & 95.369 & 51.580 & 22.662 & 41.763 & 95.574 & 51.095 & 23.454 & 42.077 \\
            \textit{Location-aware robust} & 91.665 & 43.736 & 23.163 & 38.001 & 94.642 & 52.491 & 27.526 & 43.908 & 95.670 & 56.038 & 26.375 & 44.902 & 95.959 & 55.958 & 27.105 & 45.122 \\
            \bottomrule
        \end{tabular}
        }
    }
    \label{tab:pretrainingdata}
\end{table*}

\begin{figure*}[!ht]
    \centering
    \subcaptionbox{\textit{Undefended clean precision-recall curves}}{\includegraphics[width=0.4\textwidth]{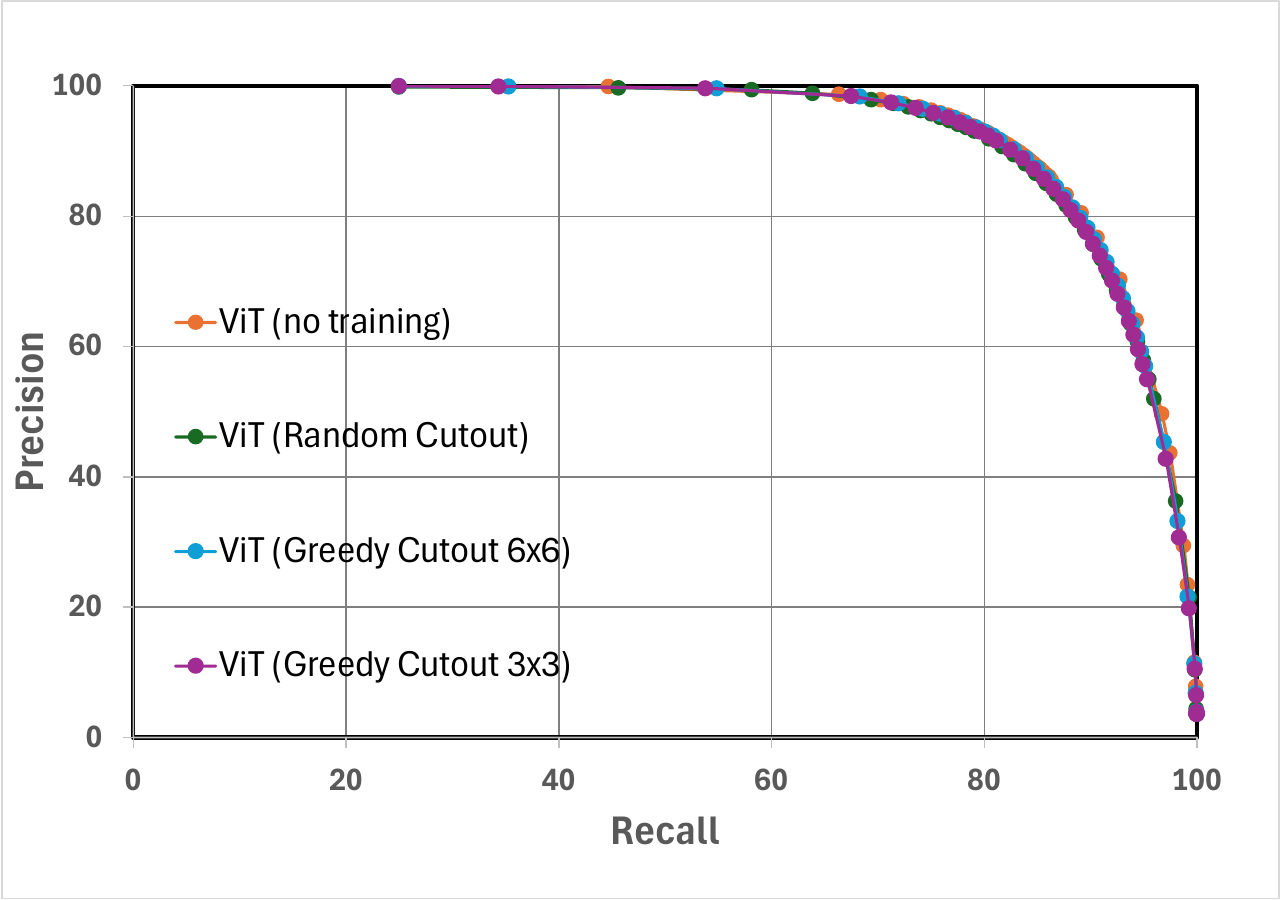}}
    \subcaptionbox{\textit{Defended clean precision-recall curves}}{\includegraphics[width=0.4\textwidth]{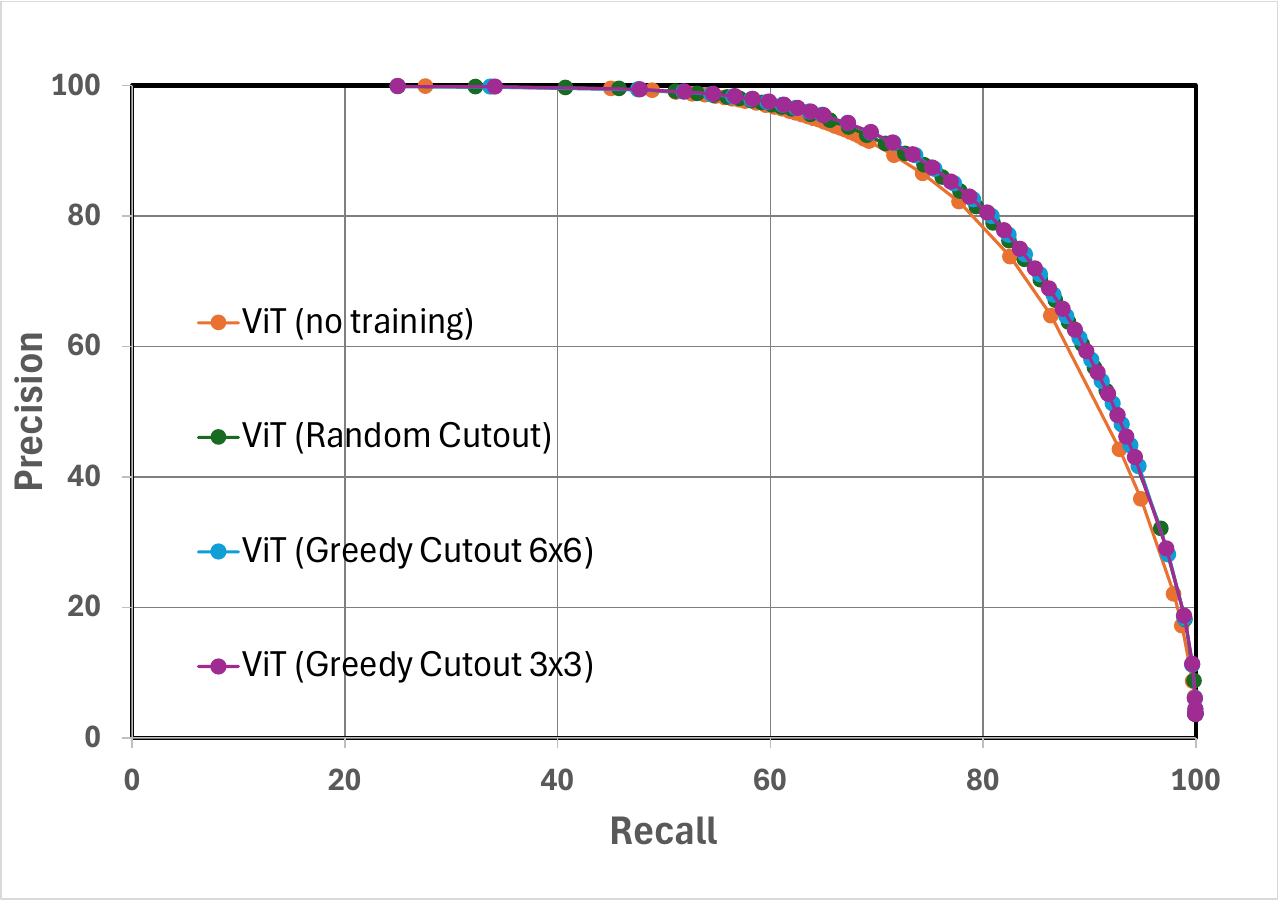}}
    \subcaptionbox{\textit{Certified robust precision-recall curves}}{\includegraphics[width=0.4\textwidth]{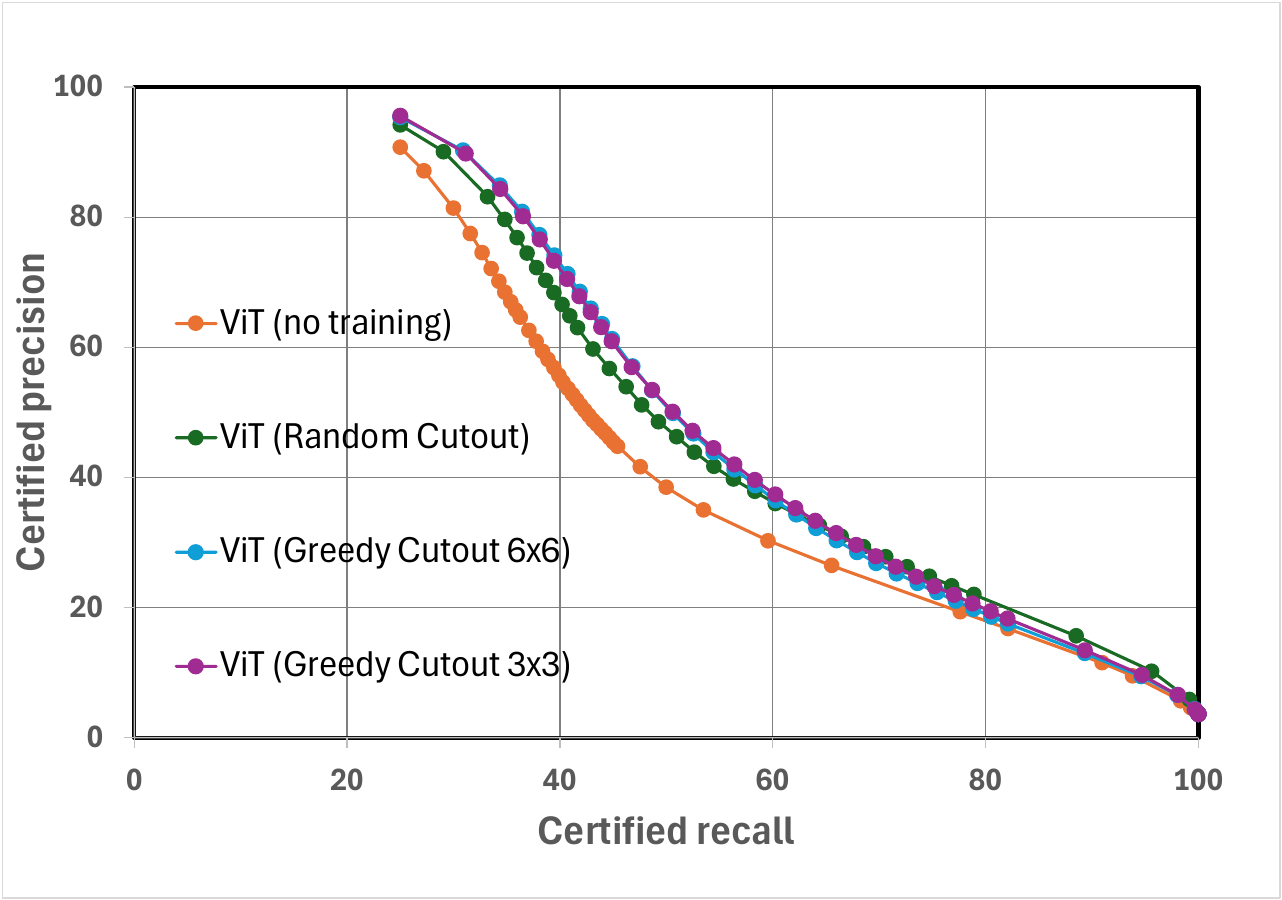}}
    \subcaptionbox{\textit{Location-aware robust precision-recall curves}}{\includegraphics[width=0.4\textwidth]{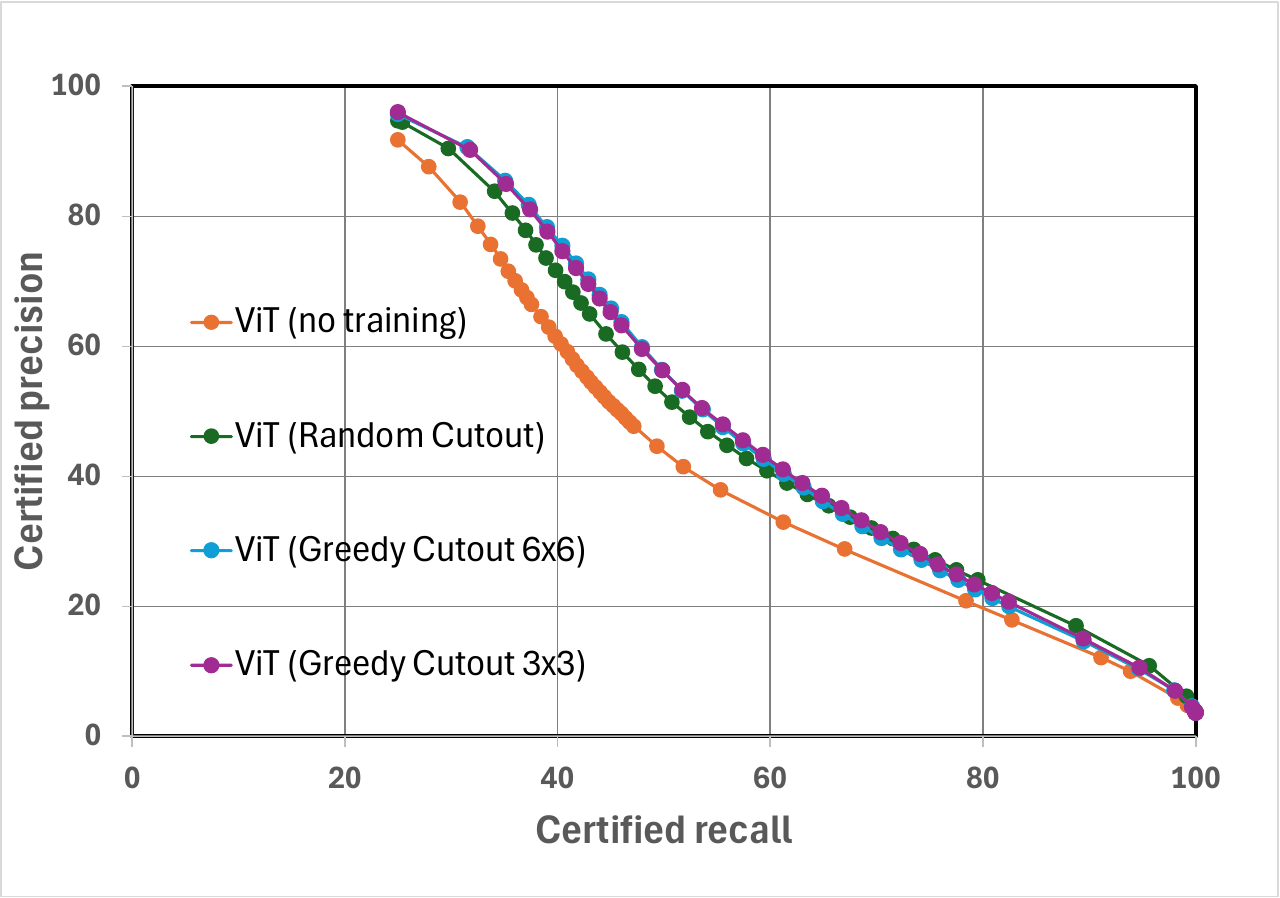}}
    \caption{\textit{PatchDEMUX precision-recall curves with ViT architecture over the MS-COCO 2014 validation dataset when using different defense fine-tuning techniques. We consider each of the four evaluation settings in separate plots. We assume the patch attack is at most $2\%$ of the image area and use a computational budget of $6 \times 6$ masks.}}
    \label{fig:pretrainingprecrecall}
\end{figure*}

Single-label CDPAs often leverage defense fine-tuning routines to improve the robustness of underlying single-label classifiers; these work by training the model on specially augmented data  \cite{devries_improved_2017, xiang_patchcleanser_2022, saha_revisiting_2023}. In this section, we investigate whether some of these routines can extend to multi-label classifiers and improve PatchDEMUX performance. We specifically consider fine-tuning strategies used by PatchCleanser, as PatchCleanser is the certifiable backbone for PatchDEMUX in this work. 

\subsection{Defense fine-tuning techniques for PatchCleanser}
\label{subsection:pretrainingpatchcleanser}
Two different defense fine-tuning techniques have been used to improve the performance of PatchCleanser: \emph{Random Cutout} \cite{devries_improved_2017} and \emph{Greedy Cutout} \cite{saha_revisiting_2023}. The former works by placing two square masks at random locations on training images, with each mask covering at most $25\%$ of the image area \cite{devries_improved_2017, xiang_patchcleanser_2022}. \citet{xiang_patchcleanser_2022} found that Random Cutout fine-tuning provides significant boosts to the robustness of PatchCleanser; intuitively, using cutout masks for defense fine-tuning helps the underlying model become more tolerant to occlusion effects from double-masking procedures. Later, \citet{saha_revisiting_2023} proposed the Greedy Cutout fine-tuning procedure and demonstrated superior performance to Random Cutout for PatchCleanser. This approach works by augmenting each training image with the pair of certification masks that greedily induce the highest loss.

\subsection{Defense fine-tuning methodology}
\label{subsection:modelpretraining}
In our experiments we compare the following three defense fine-tuning methods, which are representative of settings used in prior work \cite{devries_improved_2017, xiang_patchcleanser_2022, saha_revisiting_2023}.

\begin{itemize}
    \item Random Cutout fine-tuning with two square $25\%$ masks
    \item Greedy Cutout fine-tuning with $6 \times 6$ certification masks
    \item Greedy Cutout fine-tuning with $3 \times 3$ certification masks
\end{itemize}

\noindent
For Greedy Cutout, we compute the loss for masks while models are in evaluation mode; this approach helps avoid consistency issues associated with batch normalization. We do not consider the more complex multi-size greedy cutout approach from \citet{saha_revisiting_2023} due to difficulties with mask decompositions.

To train the model with these methods, we first obtain existing checkpoints for the MS-COCO 2014 classification task \cite{ben-baruch_asymmetric_2021, liu_query2label_2021}. We then follow the training methodology for multi-label classifiers outlined by \citet{ben-baruch_asymmetric_2021}. Specifically, we use asymmetric loss (ASL) as the loss function, a $1$cycle learning rate policy with max learning rate $\alpha_{max} = 5.0 \cdot 10^{-5}$, automatic mixed precision (AMP) for faster training, and exponential moving average (EMA) of model checkpoints for improved inference \cite{ben-baruch_asymmetric_2021}. Models are fine-tuned on copies of the MS-COCO 2014 training dataset augmented by Random Cutout and Greedy Cutout. We use the Adam optimizer for $5$ epochs, and best checkpoints are picked according to the loss on held out data\footnote{We find that fine-tuning for longer leads to overfitting.}. A cluster of NVIDIA A100 $40$GB GPUs are used to perform the fine-tuning.

\subsection{Experiment results}
\label{subsection:modelpretrainingresults}
Results for the different defense fine-tuning routines are in \cref{tab:pretrainingdata}. In addition, precision-recall plots comparing the defense fine-tuning routines for each of the four PatchDEMUX evaluation settings are present in \cref{fig:pretrainingprecrecall}. We consider the ViT architecture alone as it provides better performance compared to Resnet. Experiments are done on the MS-COCO 2014 validation dataset.

\textbf{Defense fine-tuning boosts performance.} In general, we find that using a defense fine-tuning routine of any kind leads to performance boosts for PatchDEMUX. For instance, fine-tuning helps the two certified robust evaluation settings achieve a $6-7$ point improvement in certified AP compared to the vanilla checkpoints, while the defended clean setting demonstrates a ${\sim}2$ point improvement in clean AP compared to the baseline. Greedy Cutout also provides additional robustness boosts compared to Random Cutout, with certified AP metrics being almost a full point higher; this corroborates with findings from \citet{saha_revisiting_2023}. Note that in general defense fine-tuning strategies are less effective in the clean settings. This is likely because the clean settings already demonstrate (relatively) strong performance, and thus potential gains from fine-tuning are more marginal. Nevertheless, we prioritize the defended clean setting overall as it is most representative of typical performance. The Greedy Cutout $6 \times 6$ fine-tuning strategy, which achieves the highest defended clean AP value, is therefore featured in \cref{subsection:bestresults}.

\textbf{Location-aware certification provides consistent improvements.} An interesting observation from \cref{tab:pretrainingdata} is that the location-aware robust setting provides a consistent $3$ point boost to certified AP regardless of the presence/absence of defense fine-tuning. This suggests that our location-aware certification technique has general utility across a variety of scenarios and that it ``stacks'' with other sources of robustness improvements.

\section{Runtime Analysis of PatchDEMUX}
\label{section:appendix_runtime}

\begin{table}[!ht]
    \centering
    \caption{\textit{Runtime experiments on PatchDEMUX. We report median per-sample inference time (in milliseconds) across a random sample of 2000 datapoints from the MS-COCO 2014 validation dataset. We assume the patch attack is at most $2\%$ of the image area.}}
    \resizebox{\linewidth}{!}{\begin{tabular}{lccc}
        \toprule
        \textbf{Architecture} & ViT ($2 \times 2$ masks) & ViT ($4 \times 4$ masks) & ViT ($6 \times 6$ masks) \\
        \midrule
        \textit{Undefended} & 31.130 & 31.130 & 31.130 \\
        \textit{Defended (single-label)} & 200.61 & 674.98 & 1451.1  \\
        \textit{Defended (multi-label)} & 317.30 & 1892.3 & 5668.7 \\
        \bottomrule
    \end{tabular}}
    \label{tab:runtime_vals}
\end{table}

\begin{figure}[!ht]
    \centering
    \includegraphics[width=0.4\textwidth]{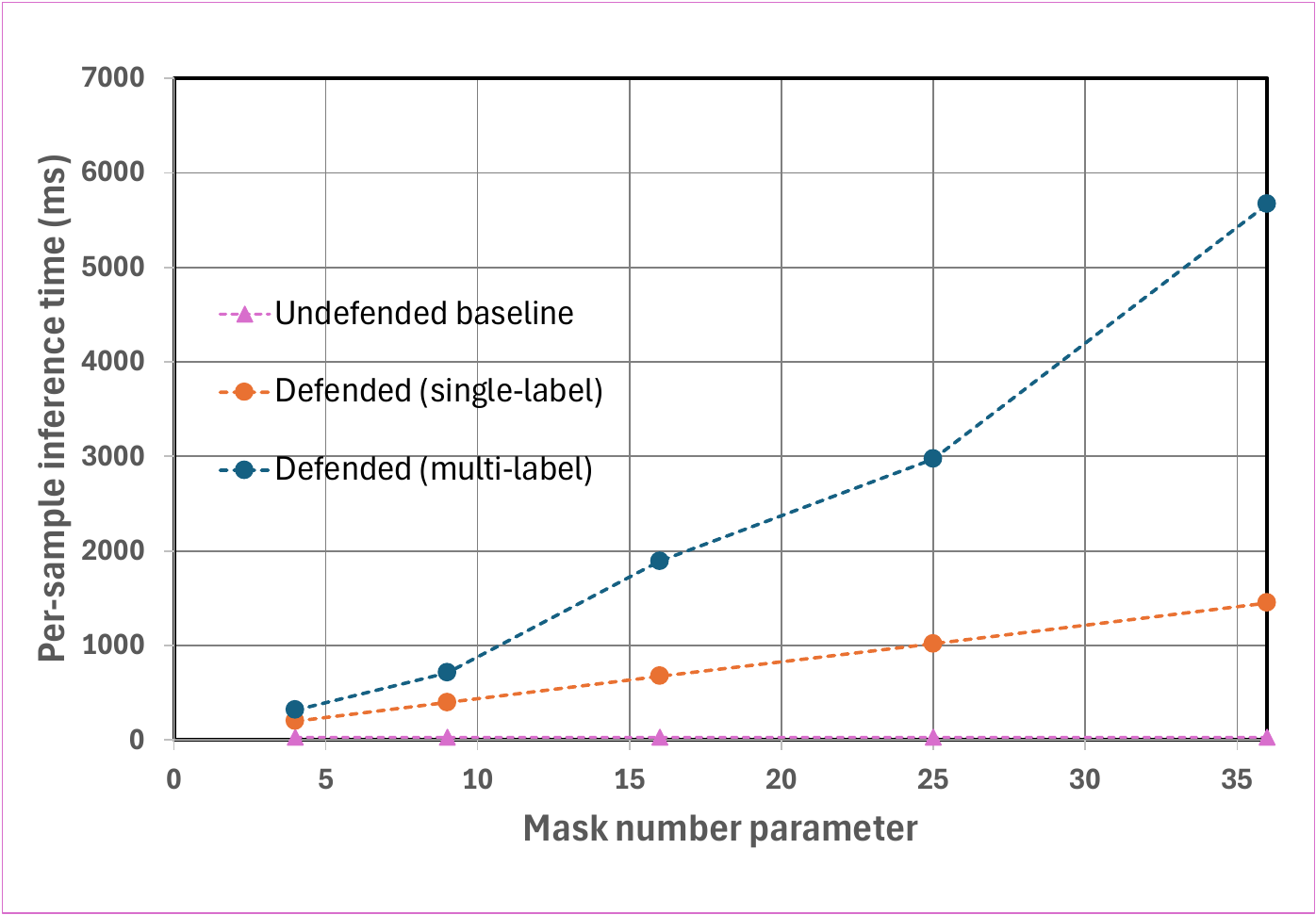}
    \caption{\textit{Plot of PatchDEMUX runtime as a function of mask number. We report median per-sample inference time (in milliseconds) across a random sample of 2000 datapoints from the MS-COCO 2014 validation dataset. We assume the patch attack is at most $2\%$ of the image area.}}
    \label{fig:runtime_fig}
\end{figure}

In this section, we analyze the runtime of the PatchDEMUX inference procedure. To determine the impact of class number, we create a restricted version of our inference procedure that operates only on the first class (i.e., it ignores the remainder of the label $\mathbf{y} \in \{0, 1\}^c$); this is essentially an instance of PatchCleanser isolated to a single class. We then track the runtime for $2000$ random datapoints from the MS-COCO 2014 validation dataset. We use the ViT checkpoints from \cref{subsection:bestresults} and use a batch size of $1$ to directly obtain per-sample inference time. The median per-sample inference times for different mask numbers are present in \cref{tab:runtime_vals} and \cref{fig:runtime_fig}.

We note that for most mask numbers the full multi-label inference procedure takes roughly $3\times$ longer than the single-label implementation. Given that MS-COCO has $c = 80$ classes, this is significantly faster than the expected runtime for the naive method from \cref{alg:inference_patchdemux}. The reason for this improvement is an implementation-level optimization that takes advantage of relatively negligible defense post-processing. Specifically, the primary bottleneck for single-label inference procedures is often model query time; the associated defense post-processing is negligible in comparison. This means that for each feedforward through the multi-label classifier we can apply single-label defense post-processing to every class and re-use individual class outputs as needed for multi-label inference. As an example, with PatchCleanser this is done by saving intermediate outputs that correspond to double-masked images. Overall, this technique helps prevent computation cost from increasing drastically with the number of classes. 

Despite this optimization, we note that the multi-label inference implementation is still not as fast as the single-label inference implementation. This is because many single-label inference procedures have worst-case scenarios which take significantly longer than typical cases. Increasing the number of classes increases the possibility that at least one class will trigger a worst-case scenario, leading to longer overall runtime.

\section{Performance on PASCAL VOC}
\label{section:appendix_pascalvoc}
\begin{table*}[!ht]
    \centering
    \caption{\textit{PatchDEMUX performance with ViT architecture on the PASCAL VOC 2007 validation dataset. Precision values are evaluated at key recall levels along with the approximated average precision. We assume the patch attack is at most $2\%$ of the image area and use a computational budget of $6 \times 6$ masks.}}
    \subcaptionbox{\textit{Clean setting precision values}\label{tab:pascalvoc_data_clean}}{
    \begin{tabular}{lcccc}
        \toprule
        \textbf{Architecture} & \multicolumn{4}{c}{ViT} \\
        \cmidrule(l){1-1}\cmidrule(l){2-5}
        \textbf{Clean recall} & $25\%$ & $50\%$ & $75\%$ & $AP$ \\
        \midrule
        \textit{Undefended} & 99.790 & 99.710 & 98.506 & 96.140 \\
        \textit{Defended} & 99.894 & 99.870 & 98.167 & 92.593 \\
        \bottomrule
    \end{tabular}
    }
    \subcaptionbox{\textit{Certified robust setting precision values}\label{tab:pascalvoc_data_robust}}{
    \begin{tabular}{lcccc}
        \toprule
        \textbf{Architecture} & \multicolumn{4}{c}{ViT} \\
        \cmidrule(l){1-1}\cmidrule(l){2-5}
        \textbf{Certified recall} & $25\%$ & $50\%$ & $75\%$ & $AP$ \\
        \midrule
        \textit{Certified robust} & 90.520 & 74.675 & 38.100 & 54.904 \\
        \textit{Location-aware} & 90.591 & 75.672 & 40.320 & 56.030\\
        \bottomrule
    \end{tabular}
    }
    \label{tab:pascalvoc_data}
\end{table*}

\begin{figure*}[!ht]
    \centering
    \subcaptionbox{\textit{Clean setting precision-recall curves}\label{fig:pascalvoc_vitbestprecrecall_clean}}{\includegraphics[width=0.4\textwidth]{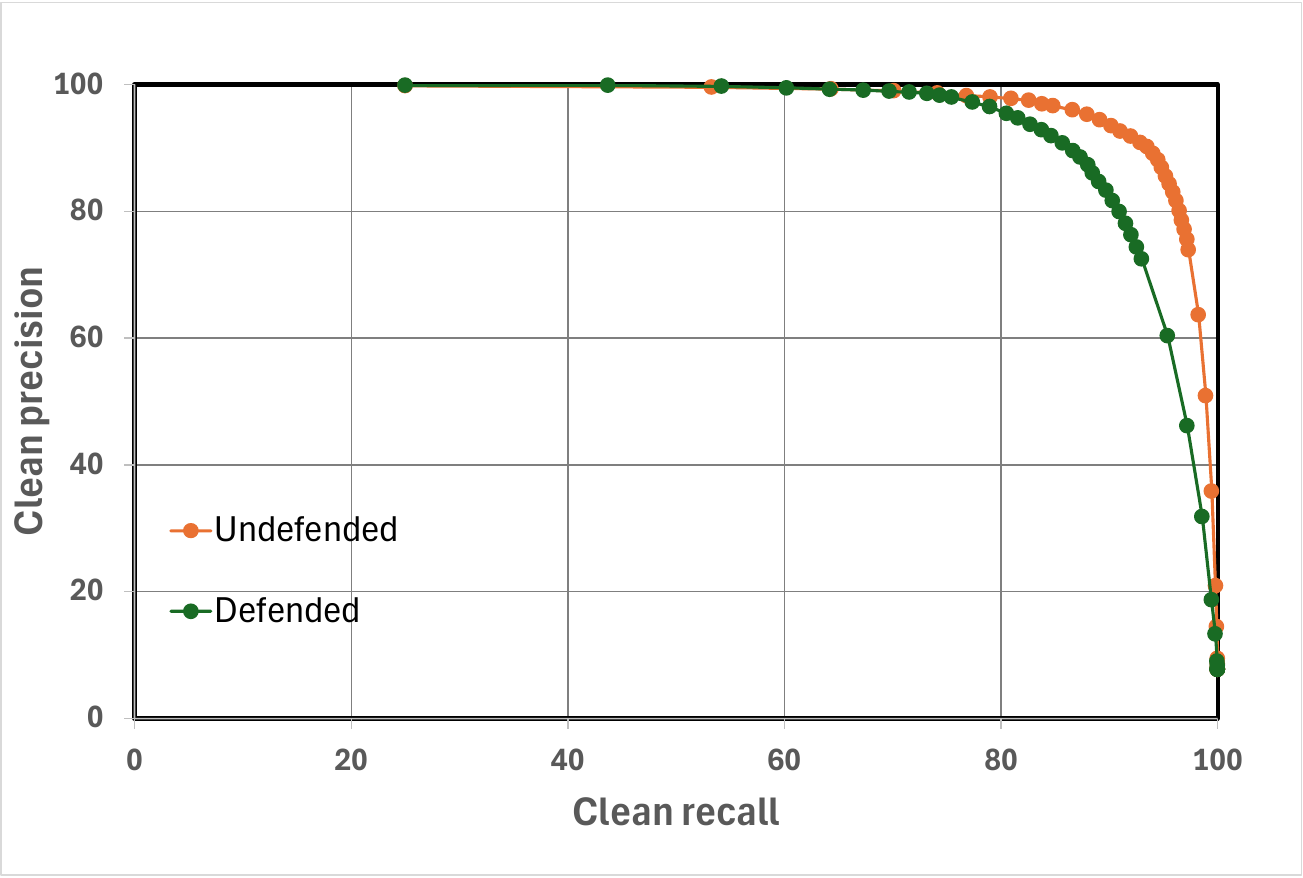}}
    \subcaptionbox{\textit{Certified robust setting precision-recall curves}\label{fig:pascalvoc_vitbestprecrecall_robust}}{\includegraphics[width=0.4\textwidth]{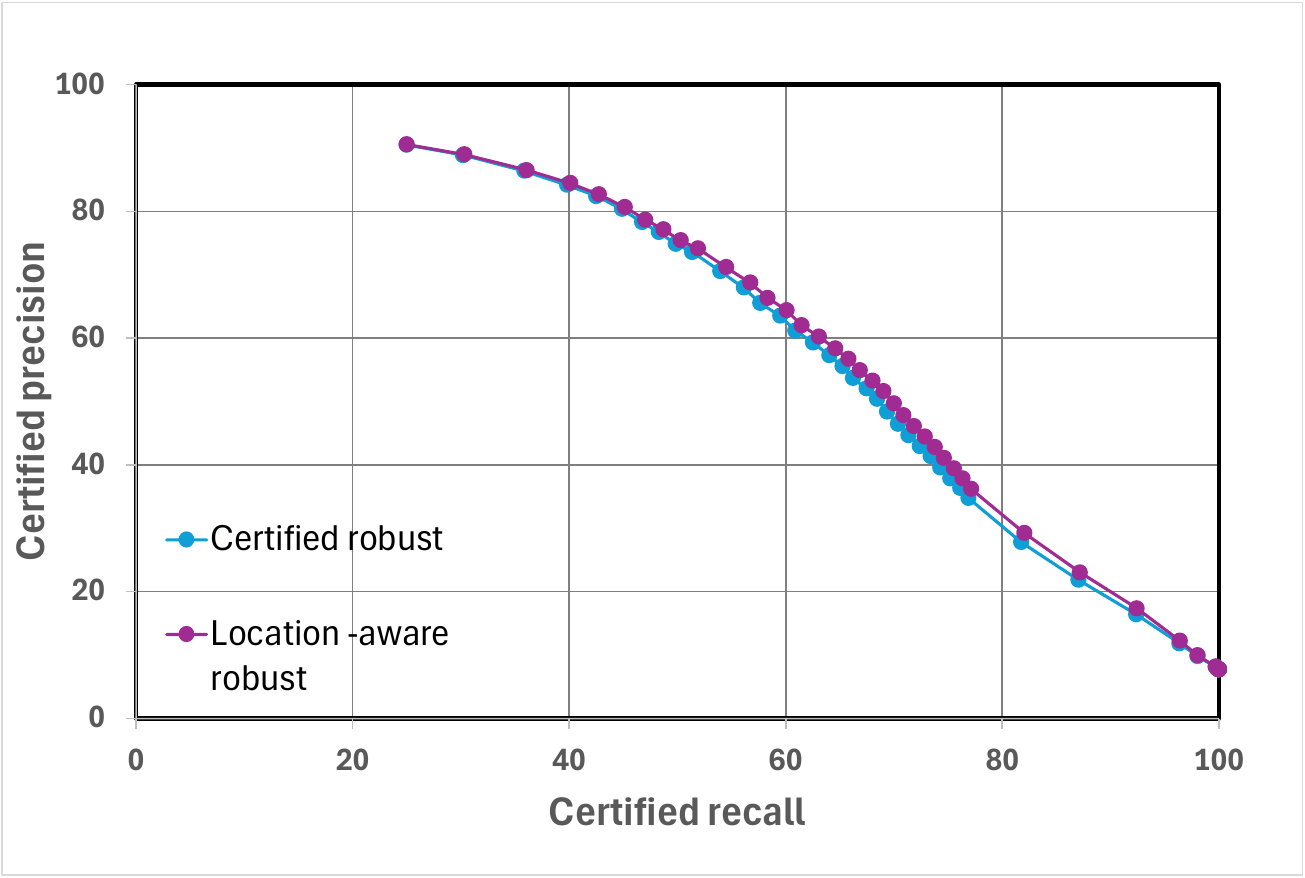}}
    \caption{\textit{PatchDEMUX precision-recall curves with ViT architecture over the PASCAL VOC 2007 test dataset. We consider the clean and certified robust evaluation settings. We assume the patch attack is at most $2\%$ of the image area and use a computational budget of $6 \times 6$ masks.}}
    \label{fig:pascalvoc_vitbestprecrecall}
\end{figure*}

In this section we report evaluation results for PatchDEMUX on PASCAL VOC. Because model checkpoints for PASCAL VOC are not readily available, we first create a multi-label classifier for the PASCAL VOC task. To do so, we use model checkpoints pre-trained on the MS-COCO dataset and fine-tune it for the PASCAL VOC dataset. We use asymmetric loss (ASL) as the loss function, a $1$cycle learning rate policy with max learning rate $\alpha_{max} = 2.0 \cdot 10^{-3}$, automatic mixed precision (AMP) for faster training, and exponential moving average (EMA) of model checkpoints for improved inference \cite{ben-baruch_asymmetric_2021}. Models are fine-tuned on the PASCAL VOC 2007 training split. We use the Adam optimizer for $15$ epochs and select the best checkpoint according to average loss on the PASCAL VOC 2007 validation split. A cluster of NVIDIA A100 $40$GB GPUs are used to perform the fine-tuning. We omit additional security fine-tuning to focus on baseline performance.

We evaluate the fine-tuned model on the PASCAL VOC 2007 test dataset. We summarize the precision values associated with key recall levels in \cref{tab:pascalvoc_data}. \cref{fig:pascalvoc_vitbestprecrecall} features precision-recall plots, while AP values are present in \cref{tab:pascalvoc_data}. We consider the ViT architecture alone as it provides better performance compared to Resnet.

\textbf{Strong all-around performance.}
As shown in \cref{tab:pascalvoc_data} and \cref{fig:pascalvoc_vitbestprecrecall}, PatchDEMUX achieves strong performance in all evaluation settings. In fact, PatchDEMUX's performance on PASCAL VOC is significantly higher than its performance on MS-COCO, with a $\sim7$ point increase in defended clean performance and $\sim12$ point increase in certified robustness metrics (see \cref{subsection:bestresults}). Overall, these stronger results are expected given that the PASCAL VOC benchmark has fewer classes than MS-COCO, making it an easier benchmark for classifiers to predict. 

\textbf{Concave robustness curves.} 
An interesting observation is that both of the PASCAL VOC robustness curves are concave in \cref{fig:pascalvoc_vitbestprecrecall_robust}. This is in contrast to MS-COCO experiments, where even after applying security fine-tuning methods the robustness curves remained convex (see \cref{fig:pretrainingprecrecall}). An important takeaway from this is that PatchDEMUX performance is dataset dependent, and robustness bounds will ultimately depend on the nature of image datapoints and/or labels. Additionally, we note that location-aware certification only provides a $\sim1$ point boost to certified AP; this suggests that location-aware certification is most beneficial when baseline robustness bounds are weak. 

\section{Tables for Security Parameter Experiments}
\label{section:appendix_sectables}

\begin{table*}[hbt!]
    \centering
    \caption{\textit{PatchDEMUX performance with ViT architecture on the MS-COCO 2014 validation dataset. We vary the mask number security parameter associated with the underlying single-label CDPA PatchCleanser and fix the estimated patch size at $2\%$ of the image area. We list even mask number values for brevity. Precision values are evaluated at key recall levels along with the approximated average precision.}}
    \subcaptionbox{\textit{Clean setting precision values}\label{tab:clean_masknum}}{
        \tabcolsep=0.15cm
        \resizebox{\linewidth}{!}{
        \begin{tabular}{lcccccccccccc}
            \toprule
            \textbf{Architecture} & \multicolumn{4}{c}{ViT ($2 \times 2$ masks)} & \multicolumn{4}{c}{ViT ($4 \times 4$ masks)} & \multicolumn{4}{c}{ViT ($6 \times 6$ masks)} \\
            \cmidrule(l){1-1}\cmidrule(l){2-5}\cmidrule(l){6-9}\cmidrule(l){10-13}
            \textbf{Recall} & $25\%$ & $50\%$ & $75\%$ & $AP$ & $25\%$ & $50\%$ & $75\%$ & $AP$ & $25\%$ & $50\%$ & $75\%$ & $AP$ \\
            \midrule
            \textit{Undefended} & 99.940 & 99.749 & 96.265 & 91.449 & 99.940 & 99.749 & 96.265 & 91.449 & 99.940 & 99.749 & 96.265 & 91.449\\
            \textit{Defended} & 99.910 & 96.999 & 75.393 & 78.727 & 99.930 & 98.845 & 83.388 & 82.529 & 99.930 & 99.138 & 85.757 & 83.776\\
            \bottomrule
        \end{tabular}
        }
    }
    \par\bigskip
    \subcaptionbox{\textit{Certified robust setting precision values}\label{tab:cert_masknum}}{
        \tabcolsep=0.15cm
        \resizebox{\linewidth}{!}{
        \begin{tabular}{lcccccccccccc}
            \toprule
            \textbf{Architecture} & \multicolumn{4}{c}{ViT ($2 \times 2$ masks)} & \multicolumn{4}{c}{ViT ($4 \times 4$ masks)} & \multicolumn{4}{c}{ViT ($6 \times 6$ masks)} \\
            \cmidrule(l){1-1}\cmidrule(l){2-5}\cmidrule(l){6-9}\cmidrule(l){10-13}
            \textbf{Certified recall} & $25\%$ & $50\%$ & $75\%$ & $AP$ & $25\%$ & $50\%$ & $75\%$ & $AP$ & $25\%$ & $50\%$ & $75\%$ & $AP$ \\
            \midrule
            \textit{Certified robust} & 41.577 & 17.924 & 9.909 & 15.735 & 87.976 & 37.163 & 19.798 & 33.231 & 90.767 & 38.490 & 20.846 & 35.003\\
            \textit{Location-aware robust} & 46.553 & 20.624 & 10.798 & 17.690 & 89.259 & 41.490 & 21.763 & 35.953 & 91.665 & 43.736 & 23.163 & 38.001\\
            \bottomrule
        \end{tabular}
        }
    }
    \label{tab:masknumvaryingdata}
\end{table*}

\begin{table*}[hbt!]
    \centering
    \caption{\textit{PatchDEMUX performance with ViT architecture on the MS-COCO 2014 validation dataset. We vary the patch size security parameter associated with the underlying single-label CDPA PatchCleanser and fix the mask number parameter at $6 \times 6$. Precision values are evaluated at key recall levels along with the approximated average precision.}}
    \subcaptionbox{\textit{Clean setting precision values}\label{tab:clean_patch_exp}}{
        \tabcolsep=0.15cm
        \resizebox{\linewidth}{!}{
        \begin{tabular}{lcccccccccccccccc}
            \toprule
            \textbf{Architecture} & \multicolumn{4}{c}{ViT ($0.5\%$ patch)} & \multicolumn{4}{c}{ViT ($2\%$ patch)} & \multicolumn{4}{c}{ViT ($8\%$ patch)} & \multicolumn{4}{c}{ViT ($32\%$ patch)}\\
            \cmidrule(l){1-1}\cmidrule(l){2-5}\cmidrule(l){6-9}\cmidrule(l){10-13}\cmidrule(l){14-17}
            \textbf{Recall} & $25\%$ & $50\%$ & $75\%$ & $AP$ & $25\%$ & $50\%$ & $75\%$ & $AP$ & $25\%$ & $50\%$ & $75\%$ & $AP$ & $25\%$ & $50\%$ & $75\%$ & $AP$ \\
            \midrule
            \textit{Undefended} & 99.940 & 99.749 & 96.265 & 91.449 & 99.940 & 99.749 & 96.265 & 91.449 & 99.940 & 99.749 & 96.265 & 91.449 & 99.940 & 99.749 & 96.265 & 91.449 \\
            \textit{Defended} & 99.947 & 99.470 & 89.150 & 85.731 & 99.930 & 99.138 & 85.757 & 83.776 & 99.907 & 97.798 & 78.712 & 80.093 & 99.529 & 89.813 & 60.543 & 69.952 \\
            \bottomrule
        \end{tabular}
        }
    }
    \par\bigskip
    \subcaptionbox{\textit{Certified robust setting precision values}\label{tab:cert_patch_exp}}{
        \tabcolsep=0.15cm
        \resizebox{\linewidth}{!}{
        \begin{tabular}{lcccccccccccccccc}
            \toprule
            \textbf{Architecture} & \multicolumn{4}{c}{ViT ($0.5\%$ patch)} & \multicolumn{4}{c}{ViT ($2\%$ patch)} & \multicolumn{4}{c}{ViT ($8\%$ patch)} & \multicolumn{4}{c}{ViT ($32\%$ patch)}\\
            \cmidrule(l){1-1}\cmidrule(l){2-5}\cmidrule(l){6-9}\cmidrule(l){10-13}\cmidrule(l){14-17}
            \textbf{Certified recall} & $25\%$ & $50\%$ & $75\%$ & $AP$ & $25\%$ & $50\%$ & $75\%$ & $AP$ & $25\%$ & $50\%$ & $75\%$ & $AP$ & $25\%$ & $50\%$ & $75\%$ & $AP$ \\
            \midrule
            \textit{Certified robust} & 97.670 & 61.867 & 30.239 & 48.820 & 90.767 & 38.490 & 20.846 & 35.003 & 44.666 & 19.249 & 11.832 & 16.961 & 6.933 & 5.827 & 4.854 & 5.297 \\
            \textit{Location-aware robust} & 97.769 & 66.350 & 32.850 & 51.158 & 91.665 & 43.736 & 23.163 & 38.001 & 50.263 & 22.965 & 13.363 & 19.713 & 9.169 & 6.997 & 5.307 & 6.195 \\
            \bottomrule
        \end{tabular}
        }
    }
    \label{tab:patchvaryingdata}
\end{table*}

In this section we provide the tables associated with the security parameter experiments in \cref{section:securityparam}. In \cref{tab:masknumvaryingdata} we list metrics associated with the mask number experiments from \cref{subsection:masknumexp}. In \cref{tab:patchvaryingdata} we list metrics associated with the patch size experiments from \cref{subsection:patchsizeexp}.

\end{document}